\documentclass[a4paper, USenglish, numberwithinsect, cleveref, autoref, thm-restate, pdfa]{lipics-v2021}

\nolinenumbers

\newif\ifarxiv
\arxivtrue

\ifarxiv \hideLIPIcs \fi 

\usepackage{mathtools}
\usepackage{bbding}

\pdfoutput=1 

\title{Separating Automatic Relations\ifarxiv\thanks{
    This is the full version of a paper accepted at MFCS'23 \cite{thispaperMFCS}.
    Licensed under \href{https://creativecommons.org/licenses/by/4.0/}{CC BY 4.0}.
}\fi%
}
\titlerunning{Separating Automatic Relations}

\author{Pablo Barcel\'o}{Institute for Mathematical and Computational Engineering, Universidad Cat\'olica de Chile \& \\ CENIA \& IMFD, Chile \and\url{https://www.pbarcelo.ing.uc.cl}}{pbarcelo@uc.cl}{https://orcid.org/0000-0003-2293-2653}{}
\author{Diego Figueira}{Univ. Bordeaux, CNRS, Bordeaux INP, LaBRI, UMR5800, F-33400 Talence, France\and
\url{https://www.labri.fr/perso/dfigueir/}}{diego.figueira@cnrs.fr}{https://orcid.org/0000-0003-0114-2257}{}
\author{Rémi Morvan}{Univ. Bordeaux, CNRS, Bordeaux INP, LaBRI, UMR5800, F-33400 Talence, France\and\url{https://www.morvan.xyz/}}{remi.morvan@u-bordeaux.fr}{https://orcid.org/0000-0002-1418-3405}{}
\authorrunning{P. Barceló, D. Figueira, and R. Morvan}

\Copyright{Barcel\'o, Figueira, Morvan}

\ccsdesc[500]{Theory of computation~Regular languages}

\keywords{Automatic relations, recognizable relations, separability, finite colorability}

\category{} 

\funding{Figueira and Morvan are partially supported by ANR QUID, grant ANR-18-CE400031. Barceló is funded by ANID - Millennium Science Initiative Program - Code ICN17002 and by 
the National Center for Artificial Intelligence CENIA FB210017, Basal ANID.}

\theoremstyle{claimstyle}

\newtheorem{fact}[theorem]{Fact}

\Crefname{claim}{Claim}{Claims}

\usepackage{dutchcal} 
\usepackage{xspace} 
\usepackage{booktabs} 
\usepackage{marginnote}

\usepackage{stmaryrd}
\newcommand{\lBrack}{\llbracket}
\newcommand{\rBrack}{\rrbracket}

\allowdisplaybreaks
\usepackage{tikz}
\usepackage{tikz-cd}

\usetikzlibrary{calc, arrows.meta, hobby, decorations.markings, patterns, positioning}

\usepackage{xcolor}

\definecolor{Desire}{HTML}{eb3b5a} 
\definecolor{Boyzone}{HTML}{2d98da} 
\definecolor{NYC Taxi}{HTML}{f7b731} 
\definecolor{Algal Fuel}{HTML}{20bf6b} 
\definecolor{Innuendo}{HTML}{a5b1c2} 

\colorlet{cBlue}{Boyzone}
\colorlet{cYellow}{NYC Taxi}
\colorlet{cGreen}{Algal Fuel}
\colorlet{cRed}{Desire}
\colorlet{cGrey}{Innuendo} 

\usepackage[xcolor, hyperref, cleveref, notion, quotation, electronic]{knowledge}
\usepackage{mathcommand}
\knowledgeconfigure{quotation, protect quotation={tikzcd}}
\knowledgeconfigure{diagnose line=true, diagnose bar=true}

\definecolor{Dark Ruby Red}{HTML}{580507}
\definecolor{Dark Blue Sapphire}{HTML}{053641}
\definecolor{Dark Gamboge}{HTML}{be7c00}

\IfKnowledgePaperModeTF{
}{
    \knowledgestyle{intro notion}{color={Dark Ruby Red}, emphasize}
    \knowledgestyle{notion}{color={Dark Blue Sapphire}}
    \hypersetup{
        colorlinks=true,
        breaklinks=true,
        linkcolor={Dark Blue Sapphire}, 
        citecolor={Dark Blue Sapphire}, 
        filecolor={Dark Blue Sapphire}, 
        urlcolor={Dark Blue Sapphire},
    }
    \IfKnowledgeElectronicModeTF{
    }{
        \knowledgeconfigure{anchor point color={Dark Ruby Red}, anchor point shape=corner}
        \knowledgestyle{intro unknown}{color={Dark Gamboge}, emphasize}
        \knowledgestyle{intro unknown cont}{color={Dark Gamboge}, emphasize}
        \knowledgestyle{kl unknown}{color={Dark Gamboge}}
        \knowledgestyle{kl unknown cont}{color={Dark Gamboge}}
    }
} 
\knowledge{typewriter, url={https://ctan.org/pkg/knowledge}}
| knowledge

\knowledge{typewriter, url={https://github.com/remimorvan/knowledge-clustering}}
| knowledge-clustering

\knowledge{url={\fullversionarxivurl}}
| full version

\knowledge{notion}
| notion@notice
| definition@notice

\knowledge{text={s.t.}, italic}
 | st

\knowledge{text={i.e.}, italic}
 | ie

\knowledge{text={w.l.o.g.}, italic}
 | wlog

\knowledge{smallcaps}
 | PSpace

\knowledge{smallcaps}
 | RE

\knowledge{smallcaps}
 | coRE

\knowledge{notion, smallcaps}
 | definability problem

\knowledge{notion, smallcaps}
 | $\REC$-definability
 | $\REC$-definability problem

 \knowledge{notion, smallcaps}
 | $\REC$-separability
 | $\REC$-separability problem

\knowledge{notion, smallcaps}
| separability problem

\knowledge{notion}
 | separates
 | separable
 | separator

\knowledge{notion}
 | support

\knowledge{notion}
 | Rational relations
 | rational relations
 | rational relation

\knowledge{notion}
 | recognizable relations
 | recognizable
 | recognizable relation
 
\knowledge{notion}
 | USLIN

\knowledge{notion}
 | identity

 \knowledge{notion}
| SLIN

\knowledge{notion}
| REC[LTT]

\knowledge{notion}
| automatic
 | automatic relations
 | automatic relation
 | Automatic relations

\knowledge{notion}
 | automatic graphs
 | automatic graph
 | automatic trees
 | automatic tree

\knowledge{notion}
 | $k$-coloring
 | 2-colorable
 | 2-coloring
 
\knowledge{notion, smallcaps}
 | $k$-colorability problem
 | $2$-colorability problem

\knowledge{notion, smallcaps}
 | finite-colorability problem

\knowledge{notion}
 | $k$-regular coloring
 | $2$-regular coloring
 | $k$-regular colorable
 | $2$-regular colorable
 | $3$-regular colorable
 | 2-regular colorable
 | 2-regular coloring
 | 3-regular colorable
 | $k$-regular colorability

\knowledge{notion}
 | regular coloring
 | regular colorable
 | regular@coloring
 | regular colorability

\knowledge{notion, smallcaps}
 | $k$-regular colorability problem
 | $2$-regular colorability problem
 
\knowledge{notion, smallcaps}
 | regular colorability problem

\knowledge{notion, smallcaps}
 | clique boundedness problem

\knowledge{notion}
 | compatible
 | incompatible

\knowledge{notion}
 | incompatibility graph
 | incompatibility relation

\knowledge{notion}
 | configuration
 | configurations

\knowledge{notion}
 | initial configuration

\knowledge{notion}
 | configuration graph

\knowledge{notion}
 | well-founded Reversible Turing Machine
 | wf-RTM

 \knowledge{notion}
 | reversible

\knowledge{notion}
 | reachable configurations

\knowledge{notion, smallcaps}
 | reachable regularity problem

\knowledge{url={https://en.wikipedia.org/wiki/De_Bruijn\%E2\%80\%93Erd\%C5\%91s_theorem_(graph_theory)}}
 | De Bruijn–Erdős Theorem

\newcommand{\fullversionarxivurl}{https://arxiv.org/abs/2305.08727}

\renewcommand{\epsilon}{\varepsilon}
\newif\ifproofappendix

\newrobustcmd\labelwithproof[1]{%
\AP\label{#1}%
\ifproofappendix%
\marginnote{\footnotesize{%
  \textnormal{First stated in page~\pageref{#1}.}%
}}
\else%
  \ifarxiv%
    \marginnote{\footnotesize{%
        \textnormal{See the proof of \Cref{#1} in page~\pageref{proof-#1}.}%
    }}%
  \fi%
\fi%
}

\newrobustcmd\introinrestatable[1]{%
\ifproofappendix%
\kl{#1}%
\else%
\intro{#1}%
\fi%
}

\newrobustcmd\introinrestatableopt[1]{%
\ifproofappendix%
\kl[#1]{#1}%
\else%
\intro[#1]{#1}%
\fi%
}

\newenvironment{proofappendix}[2]
  {
    \proofappendixtrue%
    #2*
    \proofappendixfalse%
    \begin{proof}[Proof of \Cref{#1}]
      \AP\label{proof-#1}
  }
  { 
    \end{proof}
  }

\newrobustcmd\recall[1]{
  \proofappendixtrue%
    #1*
  \proofappendixfalse%
}

\newrobustcmd\Crefapdx[1]{%
  \Cref{apdx:#1}%
}

\usepackage[backgroundcolor=orange!20, textcolor={Dark Ruby Red}, textsize=tiny, disable]{todonotes}
\definecolor{green}{RGB}{0,120,0}
\definecolor{hlyellow}{RGB}{250, 250, 190}
\definecolor{diegoeditcolor}{RGB}{210,210,255}
\definecolor{remieditcolor}{RGB}{210,255,210}

\newcommand{\sideremi}[1]{\todo[backgroundcolor=remieditcolor, size=\tiny]{{\bf R:} #1}}
\newcommand{\remi}[1]{\todo[inline,color=remieditcolor, size=\footnotesize]{{\bf R:} #1}}

\newcommand{\diego}[1]{\todo[inline,color=diegoeditcolor, size=\footnotesize]{{\bf D:} #1}}
\definecolor{light-gray}{gray}{0.9}

\newcommand{\proofcase}[1]{\noindent\colorbox{light-gray}{#1}~~}

\newrobustcmd{\wrote}{\color{wrote}\scriptsize\text{wrote}}
\newrobustcmd{\advised}{\color{advised}\scriptsize\text{advised}}

\renewcommand{\phi}{\varphi}
\renewcommand{\leq}{\leqslant}
\renewcommand{\geq}{\geqslant}
\renewcommand{\emptyset}{\varnothing}

\newcommand{\set}[1]{\{#1\}}
\newcommand{\tup}[1]{\langle#1\rangle}
\newrobustcmd{\defeq}{\mathrel{\hat{=}}}
\newcommand{\+}[1]{\mathcal{#1}}
\newrobustcmd{\N}{\mathbb{N}}
\newcommand{\dcup}{\mathrel{\dot\cup}} 
\newrobustcmd\pset[1]{\wp(#1)} 

\knowledgenewrobustcmd{\A}{\mathbb{A}} 

\knowledgenewrobustcmd{\Init}{\cmdkl{\textrm{Init}}}
\newrobustcmd{\InitC}[1]{\kl[\Init]{\color{#1}\textrm{Init}}} 
\knowledgenewrobustcmd{\Reach}{\cmdkl{\textrm{Reach}}}
\newrobustcmd{\ReachC}[1]{\kl[\Reach]{\color{#1}\textrm{Reach}}} 

\knowledgenewrobustcmd{\incompGraph}[2]{\cmdkl{\+{Inc}_{#1,#2}}}
\knowledgenewrobustcmd{\compL}{\cmdkl{\textnormal{\small(\textsc{comp}${}_\ell$)}}}
\knowledgenewrobustcmd{\compR}{\cmdkl{\textnormal{\small(\textsc{comp}${}_\textit{r}$)}}}
\knowledgenewrobustcmd{\compLpr}{\cmdkl{\textnormal{\small(\textsc{comp}${}'_\ell$)}}}
\knowledgenewrobustcmd{\compRpr}{\cmdkl{\textnormal{\small(\textsc{comp}${}'_\textit{r}$)}}}

\knowledgenewrobustcmd{\Id}{\cmdkl{\mathit{Id}}}

\knowledgenewrobustcmd{\AutGraph}[2]{\cmdkl{\langle#1,#2\rangle}}%
\knowledgenewrobustcmd{\configs}[1][T]{\cmdkl{\textit{Confs}_{#1}}}
\knowledgenewrobustcmd{\confGraph}[1][T]{\cmdkl{{\+G}^{#1}}}

\knowledgenewrobustcmd{\REC}{\cmdkl{\textsc{Rec}}}%
\knowledgenewrobustcmd\kREC[1][k]{\cmdkl{#1\textsc{-Rec}}}
\knowledgenewrobustcmd\kPROD[1][k]{\cmdkl{#1\textsc{-Prod}}}

\EventEditors{John Q. Open and Joan R. Access}
\EventNoEds{2}
\EventLongTitle{42nd Conference on Very Important Topics (CVIT 2016)}
\EventShortTitle{CVIT 2016}
\EventAcronym{CVIT}
\EventYear{2016}
\EventDate{December 24--27, 2016}
\EventLocation{Little Whinging, United Kingdom}
\EventLogo{}
\SeriesVolume{42}
\ArticleNo{23}

\newrobustcmd{\change}[1]{\textcolor{cGreen}{#1}} 

\ifarxiv\else
\relatedversiondetails{Full version}{\fullversionarxivurl}
\fi

\begin{document}

\maketitle

\begin{abstract}
    We study the separability problem for automatic relations (i.e., relations on finite words definable by synchronous automata) in terms of 
    recognizable relations (i.e., finite unions of products of regular languages). This problem takes as input two automatic relations 
    $R$ and $R'$, and asks if there exists a recognizable relation $S$ that contains $R$ and does not intersect $R'$. We show this problem to 
    be undecidable when the number of products allowed in the recognizable relation is fixed. In particular, checking if there exists 
    a recognizable relation $S$ with at most $k$ products of regular languages that separates $R$ from $R'$ is undecidable, for each fixed $k \geq 2$. 
    Our proofs reveal tight connections, of independent interest, between the separability problem and the finite coloring problem for automatic graphs, where colors are regular languages.
\end{abstract}

\medskip

{\small
\noindent
\AP
\raisebox{-.4ex}{\HandRight}\ \ This pdf contains internal links: clicking on a "notion@@notice" leads to its ""definition@@notice"".\footnote{This result was achieved by using the "knowledge" package and its companion tool "knowledge-clustering".}
}

\section{Introduction} 

\subparagraph*{Context.}
The study of classes of relations on words has become an important topic in language theory
\cite{EM,Nivat,Berstel,FS93,Choffrut-survey}, and also in areas such as databases and verification where they are used to 
build expressive languages. For instance, classes of relations of this kind are relevant for querying strings over relational 
databases \cite{BLSS03}, comparing paths in graph databases \cite{BarceloLLW12}, or defining  
string constraints for model checking \cite{LB16}. 
The most studied such classes include {\em recognizable}, {\em automatic}, and {\em rational} relations, each one of the latter two 
strictly extending 
the previous one. 
\AP
""Rational relations"" are those definable by multi-head automata, with heads possibly moving asynchronously; 
"automatic relations" are "rational relations" that are accepted by multi-head automata whose heads are forced to move synchronously; and 
"recognizable relations" correspond to finite unions of products of regular languages (or, equivalently, to languages recognized via finite monoids, by Mezei's Theorem). By definition, all of these classes coincide with the 
class of regular languages when restricted to unary relations.   

\AP
Prior work has focused on the ""$\REC$-definability problem"", 
which takes as input an $n$-ary "rational relation" $R$ and asks
whether it is equivalent to a "recognizable relation"
$\bigcup_i L_{i,1} \times \cdots \times L_{i,n}$, where each $L_{i,j}$ is a regular language. 
Intuitively, the problem asks whether the different components of  
the "rational relation" $R$ are almost independent of one another.
The study of "$\REC$-definability"
is relevant since relations enjoying this property are often amenable to some analysis including, e.g.,
abstract interpretations in program verification, variable elimination in constraint logic programming, and 
query processing over constraint databases (see the introduction of \cite{ICALP19} for a thorough discussion on this topic).

In general, "$\REC$-definability" of "rational relations" is undecidable, but it becomes decidable for two important subclasses: 
deterministic "rational relations" and "automatic relations". 
For deterministic "rational relations",  "$\REC$-definability" has been shown to be decidable in double-exponential time for binary relations by Valiant \cite{Valiant75}---improving Stearns's triple-exponential bound \cite{Stearns}. The decidability result was later extended to relations of arbitrary arity by Carton, Choffrut and Grigorieff~\cite[Theorem 3.7]{CCG06}.
For "automatic relations", the decidability of "$\REC$-definability"
can be obtained by a  simple reduction to the problem of checking whether a finite automaton recognizes an infinite language~\cite{DBLP:journals/dmtcs/LodingS19} -- which is decidable via a standard reachability argument. 
The precise complexity of the problem, however, was only recently pinned down. By applying 
techniques based on Ramsey Theorem over infinite graphs, it was shown that "$\REC$-definability" of "automatic relations" is 
"PSpace"-complete when
 relations are specified by non-deterministic automata \cite[Theorem 1]{ICALP19} \cite[Corollary 2.9]{BGLZ22}.

\AP On the other hand, much less is known about the ""$\REC$-separability problem"", which takes two $n$-ary "rational relations" 
$R,R' \subseteq \A^* \times \A^*$ and checks whether there is a "recognizable relation" 
$S = \bigcup_i L_{i,1} \times \cdots \times L_{i,n}$
with $R \subseteq S$ and $R' \cap S = \emptyset$. In other words, this problem asks whether we can {\em overapproximate} $R$ with a recognizable 
relation $S$ that is constrained not to intersect with $R'$. Separability problems of this kind abound in theoretical computer science, in particular
in formal language theory where they have gained a lot of attention over the last few years ---see, e.g., \cite{PZ14,K16,CMRZZ17,CCLP17}. 

\remi{The following paragraph is not very clear.}
As for definability, the "$\REC$-separability problem" for "rational relations"
is in general undecidable.
In this paper we focus on the separability problem for "automatic relations", that is, the restriction of the "$\REC$-separability
problem" defined above to the case when both $R$ and $R'$ are "automatic relations". 
Notice that when $R'$ is the complement of $R$ this problem boils down to "$\REC$-definability". 
However, "$\REC$-separability" for "automatic relations" is more general than "$\REC$-definability", and to this day it is unknown whether 
it is decidable. 

\subparagraph*{Main contributions and technical approach.}
While we do not solve the separability problem for "automatic relations", 
we report on some significant progress in our understanding of the problem.
We start by establishing a tight connection between "$\REC$-separability" and the colorability problem 
for ``"automatic graphs"'', which may shed some light on the difficulty of the former problem.
An "automatic graph" \cite{BG00,IKR02,DBLP:conf/ifipTCS/KuskeL08,KL10}\sideremi{I don't understand these citations, what are they doing here?} is an infinite graph defined on a regular set of finite words, 
whose edge set is described by a binary "automatic relation". The "regular colorability problem" is then the problem of checking if a given "automatic graph" is finitely colorable, with the restriction that each color forms a regular language.
Concretely, we show that the "$\REC$-separability problem" for binary "automatic relations" is equivalent, under polynomial time reductions, to the "regular colorability problem".
\AP Moreover, we introduce a hierarchy $(\kREC)_{k > 0}$ of "recognizable relations"
so that the coloring problem, when restricted to $k>0$ colors---called "$k$-regular colorability problem"--- reduces to the separability problem by relations of $\kREC$. Concretely:

\begin{restatable*}{theorem}{regcolorabilityequivseparability}
    \AP\label{thm:reg-colorability-equiv-separability}
    There are polynomial-time reductions: 
    \begin{enumerate}
        \item from the "$\REC$-separability problem" to the "regular colorability problem"; 
        \item from the "regular colorability problem" to the "$\REC$-separability problem"; and
        \item from the "$k$-regular colorability problem" to the $\kREC$-"separability problem", for every $k > 0$.
    \end{enumerate}
    Further, the last two reductions are so that the second relation in the instance of the "separability problem" is the identity $\Id$.
\end{restatable*}
The "regular colorability problem" seems challenging, and in particular we lack tools for establishing that an "automatic graph" is finitely colorable; let alone checking that said colors define regular sets. 
%
On the other hand, it is easy to see that the "$k$-regular colorability problem"
is undecidable for each fixed $k > 1$ if we lift the restriction that colors define regular sets, i.e., 
checking if an "automatic graph" admits a $k$-coloring---this has been proved in an unpublished thesis by 
K\"ocher \cite[Proposition 6.5]{Kocher}.
To be more precise, the problem is even co-recursively enumerable-complete\footnote{The upper bound
follows from "De Bruijn–Erdős Theorem".}.
We establish that this undecidability holds even with the restriction on colors being regular sets:     

\begin{restatable*}{theorem}{kregcolundec}
    \AP\label{thm:k-reg-col-undec}
    The {\sc"$k$-regular colorability problem"} on "automatic graphs" is undecidable, for every $k\geq 2$. More precisely, the problem is recursively enumerable-complete. This holds also for connected "automatic graphs".
\end{restatable*}

Note that the definitions of "$k$-regular colorability problem"
and "$k$-colorability problem" look similar, and are both undecidable,
but the former is "RE"-complete while the latter is "coRE"-complete.

\AP
By reduction from the "$k$-regular colorability problem" we obtain an important consequence for our separability problem:  
It is undecidable to check if two "automatic relations" can be separated by a "recognizable relation"
defined by a {\em fixed} number of unions of products of regular languages. More specifically, fix $k > 0$ 
and define $\intro*\kPROD$ as the class of "recognizable relations" of the form $S = \bigcup_{1 \leq i \leq k} L_{i,1} \times \cdots \times L_{i,n}$---this hierarchy is intertwined with
the $(\kREC)_{k > 0}$ hierarchy introduced previously. We 
show that the $\kPROD$-"separability problem", "ie", the problem of checking separability for binary "automatic relations" $R$ and $R'$ in terms of a recognizable relation $S$ in the class 
$\kPROD$, is undecidable for any $k \geq 2$.
\begin{restatable*}{theorem}{kprodundecidable}
    \AP\label{thm:kprod-undecidable}
    The $\kPROD$-"separability problem" is undecidable, for every $k \geq 2$.
\end{restatable*}

At this point, a natural question is whether our choice of restricting the study to the class $\kPROD$, for fixed $k > 1$, is not too strong, in the sense that 
it turns undecidable not only the separability but also the \emph{definability} problem for "automatic relations". We show that this is not the case; in fact, by using a 
simple adaptation of the proof techniques in \cite{ICALP19} we can show that the problem of checking if an automatic relation can be expressed 
as a relation in $\kPROD$, for any fixed $k > 0$, is decidable in single-exponential time:
\begin{restatable*}{corollary}{kproddefdec}
    The $\kPROD$-"definability problem" is decidable, for every $k > 0$.
\end{restatable*}

\subparagraph*{Remark.}
For simplicity, we focus on binary "automatic relations" only. Extending the decidability results to $n$-ary automatic relations, for $n > 2$ is direct by applying tools in \cite{ICALP19}. 

%
%
%

\section{Preliminaries}

\subparagraph*{Automatic and recognizable relations.}
Let $\A$ be a finite alphabet. We write $\A_\bot$ for the 
extension of $\A$ with a fresh symbol $\bot$. 
Given a pair $(w_1,w_2) \in \A^* \times \A^*$, 
we write $w_1 \otimes w_2$ for the word over alphabet 
$\A_\bot \times \A_\bot$ that is obtained as follows: first, padding the shorter word with $\bot$'s until both words are of the same length, 
and then reading the two words synchronously as if they were a single word over a binary alphabet. For example, if $w_1 = aaba$ and $w_2 = aa$, then 
$w_1 \otimes w_2 = (a,a)(a,a)(b,\bot)(a,\bot)$. 
For any relation $R \subseteq \A^* \times \A^*$, let us write
\knowledgenewrobustcmd{\lotimes}[1]{\cmdkl{\otimes}#1}%
\AP
$\intro*\lotimes R$ to denote the set 
\[\lotimes R \ \defeq \ \set{u \otimes v \mid (u,v) \in R} \ \subseteq \ (\A_\bot \times \A_\bot)^*.\]
We then have the following:  
\AP
\begin{itemize} 
\item 
$R \subseteq \A^* \times \A^*$ is an ""automatic relation"" if{f} $\lotimes R$ is a regular language;
\item $R \subseteq \A^* \times \A^*$ is a ""recognizable relation"" if{f}
$R = \bigcup_{i=1}^n A_i \times B_i$,
where $n \in \N$ and all the $A_i$'s and $B_i$'s are regular languages over $\A$. 
\end{itemize} 
\AP
We denote by $\intro*\REC$ the class of all "recognizable relations".

\begin{example} 
For any fixed constant $c > 0$, the relation $R$ composed by all pairs of words of the form $(a^n,a^{n+c})$, for $n \geq 0$, is "automatic". In turn, $R$ is not "recognizable". An example of a non-"automatic relation" is the one consisting of all pairs of the form $(a^n,a^{d \cdot n})$, for $n > 0$, for any constant $d > 1$.  \qed
\end{example} 

\subparagraph*{Separability.}\AP
Let $R$ and $R'$ be "automatic relations" over an alphabet $\A$. A "recognizable relation" $S$ over $\A$ 
""separates"" $R$ from $R'$ 
if $R \subseteq S$ and $R' \cap S = \emptyset$. 

\begin{example} 
Consider the "automatic relations" $R = \{(a^n,a^{n+1}) \mid n\geq 0\}$ and $R' = \{(a^n,a^{n+2}) \mid n\geq 0\}$. 
They are "separable" by the "recognizable relation" 
\[S = (A_{\rm even} \times A_{\rm odd}) \, \cup \, (A_{\rm odd} \times A_{\rm even}),\] 
where $A_{\rm even}$ and $A_{\rm odd}$ are the regular languages $(aa)^*$ and $a(aa)^*$, respectively. \qed 
\end{example}

We study the following separability problem, for a class $\cal C$ of "recognizable relations". 
\AP
\begin{center}
\fbox{\begin{tabular}{rl}
{\textbf{Problem}}: & $\cal C$-""separability problem"" \\
{\textbf{Input}}: & "Automatic relations" $R$ and $R'$ over $\A$ \\   
{\textbf{Question}}: & Is there a "recognizable relation" in $\cal C$ over $\A$ 
that "separates" $R$ from $R'$? 
\end{tabular}} 
\end{center}   
\AP
We also consider the $\cal C$-""definability problem"", 
which takes as input an automatic relation $R$ and asks if there is a recognizable relation $S$ in $\cal C$ with $S = R$. 
It is easy to see that the $\cal C$-"definability problem" corresponds to an instance
of the $\cal C$-"separability problem". 

\begin{fact}
    For any class $\cal C$ of recognizable relations, the $\cal C$-"definability problem" is Turing-reducible to the $\cal C$-"separability problem".
\end{fact}
\begin{proof}
    Reduce an instance $R$ of the "definability problem" to the instance
    $(R,\, (\A^*\times \A^*) \setminus R)$ of the "separability problem".
\end{proof}

The following is known regarding the complexity of the $\REC$-definability problem. 

\begin{proposition} {\em \cite[Theorem 1]{ICALP19}} 
The $\REC$-"definability problem" for "automatic relations" specified by non-deterministic automata 
is "PSpace"-complete. 
\end{proposition}  

\subparagraph*{Automatic graphs.}\AP 
Let $L$ be a language of finite words over $\A$, and $R \subseteq L \times L$ binary relation over $\A$. They naturally define a directed graph 
$G=\intro*\AutGraph{L}{R}$, "ie", 
the nodes of $G$ are the words over $L$ and there is an edge in $G$ from word $u$ to word $v$ if{f} $(u,v) \in R$. 
An ""automatic graph"" is a graph of the form $\AutGraph{L}{R}$, for $R$ an
"automatic relation" and $L$ a regular language\footnote{Note that an "automatic graph" can contain self-loops.
However, since the
presence of such an edge prevent the graph from being "$k$-colorable" for any $k \geq 0$,
all our examples will be self-loop-free.}.
A ""$k$-coloring"" of $\AutGraph{L}{R}$ is a partition of $L$
into $k$ sets $V_1, \dotsc V_k$ such that $(V_i \times V_i) \cap E = \emptyset$ for every $i$. 

\begin{example} \AP\label{ex:fc} 
Consider again the "automatic relation" $R = \{(a^n,a^{n+c}) \mid n \geq 0\}$, where $c > 0$ is a fixed constant. The graph $\AutGraph{a^*}{R}$ is formed by a disjoint union of $c$ infinite directed paths, and thus it is "2-colorable". \qed 
\end{example} 
\AP
A ""$k$-regular coloring"" of an "automatic graph" is a "$k$-coloring" whose colors $(V_i)_{1 \leq i \leq k}$ are regular languages. 
\AP
A ""regular coloring"" is a "$k$-regular coloring" for some $k$.


\begin{example} 
The "automatic graph" $\AutGraph{a^*}{R}$ from \Cref{ex:fc} is "2-regular colorable". In fact, it suffices to define color $V_1$ as having every word of the form $a^n$ with 
$n \equiv i$ ({\rm mod} $2c$), for $i \in [0,c-1]$, and $V_2 = \A^* \setminus V_1$.\qed 
\end{example} 

\AP The ""$k$-regular colorability problem"" is the problem of whether a given "automatic graph" has a "$k$-regular coloring". \AP
The ""regular colorability problem"" is the problem of whether a given "automatic graph" has a "regular coloring".

\section{Separability is Equivalent to Regular Colorability}

We start by showing that the "separability problem" in terms of arbitrary "recognizable relations" is equivalent, under polynomial time reductions, 
to the "regular colorability problem". To make our statement precise, we need some terminology introduced below. 
\AP
Let $\intro*\kREC$ be the class of languages expressed by unions of products of $k$ regular languages which form a partition, that is (in the binary case), relations of the form $(L_{i_1} \times L_{j_1}) \cup \dotsb \cup (L_{i_\ell} \times L_{j_\ell})$, with $i_1,j_1,\hdots,i_\ell, j_\ell \in \lBrack 1,k \rBrack$, for some regular partition $L_1, \dotsc, L_{k}$ of $\A^*$ and $\ell \in \N$.
\AP
Note that $\REC = \bigcup_k \kREC$.
\AP%
Let us denote by $\intro*\Id$ the identity relation (on any implicit alphabet). Observe that $\Id$ is "automatic" but not "recognizable".

\regcolorabilityequivseparability


   
\begin{proof}
   We start with the last two reductions.
    Given an "automatic graph" $\AutGraph{L}{E}$ over an alphabet $\A$, consider the instance $R_1,R_2$ for the $\REC$-"separability problem", where 
    $R_1 = E$ and $R_2 = \Id$. 
    If $\AutGraph{L}{E}$ is "$k$-regular colorable" via the coloring $V_1, \dotsc, V_k$ then the $\kREC$ relation
    $\bigcup_{i \neq j} V_i \times V_j$ separates $R_1$ and $R_2$.
    Conversely, if a $\kREC$ relation $R \subseteq \A^* \times \A^*$ on the partition $V_1 \dcup \dotsb \dcup V_k = \A^*$ separates $R_1$ and $R_2$, then $\bigcup_{i \neq j} V_i \times V_j$ also separates $R_1$ and $R_2$, and this implies that $V_1, \dotsc, V_k$ is a $k$-coloring for $\AutGraph{\A^*}{E}$.

\AP For the first reduction, let us introduce some terminology.
Given two relations $R_1,R_2$ over $\A^*$, say that $u \in \A^*$ is ""compatible"" with
$u' \in \A^*$ when for all words $v \in \A^*$:
\begin{center}
    \intro*\compL: $(u,v) \in R_1 \Rightarrow (u',v) \not\in R_2$%
    ,\hphantom{\quad and \quad}
    \intro*\compR: $(v,u) \in R_1 \Rightarrow (v,u') \not\in R_2$,\\
    \intro*\compLpr: $(u',v) \in R_1 \Rightarrow (u,v) \not\in R_2$%
    \hphantom{,}\quad and \quad
    \intro*\compRpr: $(v,u') \in R_1 \Rightarrow (v,u) \not\in R_2$.
\end{center}
\AP
Define the ""incompatibility graph"" $\intro*\incompGraph{R_1}{R_2}$
as the graph whose vertices are all words of $\A^*$,
and with an edge from $u$ to $v$ whenever $u$ is not "compatible" with $v$.
Note that $\incompGraph{R}{\Id}$ is exactly the graph $\AutGraph{\A^*}{R}$.
For a less trivial example of an "incompatibility graph", see
\ifarxiv
 \Cref{apdx-sec:incompatibility}.
\else
 the "full version".
\fi

\begin{restatable}{lemma}{incompisautomatic}
    \AP\labelwithproof{lem:incomp-is-automatic}
    If $R_1$ and $R_2$ are "automatic", then so is $\incompGraph{R_1}{R_2}$.
    Moreover, we can build an automaton for $\incompGraph{R_1}{R_2}$ in polynomial time in the size of the automata for $R_1$ and $R_2$.
\end{restatable}


    \AP Given an instance $(R_1,R_2)$ of the "separability problem", we
    reduce it to the "regular colorability problem" on its "incompatibility graph" $\incompGraph{R_1}{R_2}$.

    \proofcase{Left-to-right implication:} 
   Assume that there exists $S$ in $\kREC$ that "separates" $R_1$ from $R_2$.
    Then $S$ can be written as $(A_{i_1}\times A_{j_1}) \cup \cdots \cup (A_{i_\ell}\times A_{j_\ell})$, 
    where $(A_1,\hdots,A_k)$ is a partition of $\A^*$ in $k$ regular languages.
   We define the color of a word $u \in \A^*$ as the unique $i \in \lBrack 1,k \rBrack$
    "st" $u \in A_i$. In other words, the coloring is simply $(A_1,\hdots,A_k)$. 

    This is indeed a proper coloring: if $u$ and $u'$ have the same color,
    we claim that $u$ is "compatible" with $u'$. Indeed, take any $v \in \A^*$: if $(u,v) \in R_1$,
    then $(u,v) \in S$, so $(u,v) \in A_{i_m}\times A_{j_m}$ for some $m$. But since $u$ has the same color 
    as $u'$, the fact that $u \in A_{i_m}$ implies $u' \in A_{i_m}$, and hence 
    $(u',v) \in A_{i_m}\times A_{j_m}\subseteq S$.
    But $S$ separates $R_1$ from $R_2$, and therefore $(u',v) \not\in R_2$. This tells us that \compL\ holds. 
    The other conditions hold by symmetry.
    We conclude that $(A_1,\hdots,A_k)$ defines
    a proper coloring of $\incompGraph{R_1}{R_2}$, and this coloring, with $k$ colors, is 
    "regular@regular coloring" since the $A_i$'s are regular languages by definition.

    \proofcase{Right-to-left implication:} Assume that $\incompGraph{R_1}{R_2}$ is finitely colorable, say by
    $(A_1,\hdots,A_k)$. Then let $S$ be the union of all $S_i$'s where
    \begin{align*}
        S_{i} & \defeq \{(u,v) \mid u \in A_i \text{ and } (u',v) \in R_1 \text{ for some } u' \in A_i \}\\
&         \hphantom{\defeq~}\cup \{(u,v) \mid v \in A_i \text{ and } (u,v') \in R_1 \text{ for some } v' \in A_i \}.  
    \end{align*}
    Since $(A_1,\hdots,A_k)$ covers every node of $\incompGraph{R_1}{R_2}$, we get $R_1 \subseteq S$.
    Moreover, we claim that $R_2 \cap S = \emptyset$. Indeed, if $(u,v) \in S$,
    then $(u,v) \in S_{i}$ for some $i,j$. It either means that \proofcase{1}
    $(u',v) \in R_1$ for some $u' \in A_i$, or \proofcase{2} $(u,v') \in R_2$
    for some $v' \in A_i$. In case \proofcase{1}, the fact that $u \in A_i$ implies that $u$ and $u'$
    have the same color. Thus, $u$ must be "compatible" with $u'$ and hence
    $(u,v) \not\in R_2$ using \compLpr. The other case is symmetric.
    Therefore, $(u,v) \not\in R_2$, and thus $S$ separates $R_1$ from $R_2$.

    Finally,  $S$ is "recognizable"; in fact, 
    $S = \bigcup_{i=1}^k \bigl( A_i \times R_1[A_i] \bigr) \cup \bigl( R_1^{-1}[A_i] \times A_i \bigr)$, 
    where for any set $X \subseteq \A^*$ we define $R_1[X]$ (resp. $R_1^{-1}[X]$) as the set
    of $v\in \A^*$ (resp. $u \in \A^*$) such that $(u,v) \in R_1$ for some $u \in X$
    (resp. $v\in X$).
    Hence, $R_1$ and $R_2$ are $\REC$-"separable". 
\end{proof}

It is not known to date whether the "regular colorability problem" is decidable, and hence 
the same holds for the $\REC$-"separability problem"
in light of the previous theorem. This is due to the fact that there are no known characterizations of when an "automatic graph" is finitely colorable. 
In spite of this, we believe that the connection between separability and finite colorability is of interest, as it provides us with a way to define and study meaningful 
restrictions of our problems. The first such restriction corresponds to the "$k$-regular colorability problem" for "automatic graphs", which we study in the next section. 


\section[$k$-Regular Colorability Problem]{$\boldsymbol{k}$-Regular Colorability Problem}

While we do not know how to approach the "regular colorability problem", we show that as soon as we add the restriction that the number of colors is bounded, the problem becomes undecidable; i.e., the "$k$-regular colorability problem" is undecidable for $k\geq 2$. Using this, we obtain in the next section 
the undecidability for the "separability problem" on two natural classes of 
"recognizable relations". This is proven by a reduction from a suitable problem on reversible 
Turing Machines with certain restrictions, which we call ``well-founded''.

\subsection{Regularity of Reachability for Turing Machines}
We use the standard notation $u[i..j]$ to denote the factor of a word $u$ between (and including) positions $i$ and $j$, and $u[i]$ to denote $u[i..i]$.
Consider any deterministic Turing Machine (TM) $T = \tup{Q,\Gamma,\bot,\delta,q_0,F}$, where $Q$ is the set of states, $\Gamma$ is tape alphabet, $\bot$ is the blank symbol, $\delta: (Q \setminus F) \times \Gamma_\bot \to Q \times \Gamma \times \set{L,R}$ is the transition (partial) function, where $\Gamma_\bot = \Gamma \cup \set\bot$, and $q_0$ and $F$ is the initial and set of final states, respectively.
We represent a configuration with tape content $w \cdot \bot^\omega$ (where $w \in \Gamma^* \cdot \set{\bot}$), in state $q$ and with the head pointing to the cell number $1 \leq i \leq |w|$, as the string
\[
   w[1..i-1] \cdot (w[i],q) \cdot w[i+1..|w|] 
\]
over the alphabet $\A_T = \Gamma \cup (\Gamma_\bot \times Q)$.
\AP In light of this representation, we will henceforth denote by ``configuration'' any string from the set  $\intro*\configs \defeq (\Gamma^* \cdot (\Gamma_\bot \times Q)) \cup  (\Gamma^* \cdot (\Gamma \times Q) \cdot \Gamma^*)$. The ""initial configuration"" is $(\bot,q_0)$.
\AP The ""configuration graph"" of $T$ is the infinite graph $\intro*\confGraph$ having $\configs$ as set of vertices and an edge from $c$ to $c'$, denoted $c \rightarrow c'$, if $c'$ is the configuration of the next step of $T$ starting from $c$. Observe that the "configuration graph" $\confGraph$ of any TM $T$ is an effective "automatic graph" (see, e.g., \cite{KL10}).

\AP We say that a deterministic TM $T$ is ""reversible"" if every node of $\confGraph$ has in-degree at most 1, in other words if the machine is co-deterministic\footnote{Note
that a modern proof of undecidability of the isomorphism problem for automatic structures by
Blumensath \cite[\S VIII. Theorem 4.3, p. 396 \& second claim, p. 398]{blumensath2023MSO} also relies on the use of "reversible" Turing machines.}.
We say that a TM $T$ is a ""well-founded Reversible Turing Machine"" (\reintro{wf-RTM}) if its "configuration graph" is such that (1) the "initial configuration" has in-degree 0 (2) every node has in-degree and out-degree at most one (3) there are no infinite backward paths $c_1 \leftarrow c_2 \leftarrow \dotsb$ in $\confGraph$. 

\AP Note that every "well-founded Reversible Turing Machine" is deterministic and "reversible" and, moreover,
its "configuration graph" is a (possibly infinite) disjoint union of
directed paths, which are all finite, or isomorphic to $(\mathbb{N}, +1)$.
The set of ""reachable configurations"", denoted by $\intro*\Reach$, is 
the set of all configurations that admit a path from the "initial configuration"
in $\confGraph$, for a given TM $T$.
Such a configuration graph is depicted on \Cref{subfig:config-graph-wf-RTM}.

\AP The ""reachable regularity problem"" is the problem of, given a "wf-RTM" $T$, whether its set of "reachable configurations" is a regular language. To show that is it undecidable, we exhibit a reduction from the halting problem on deterministic "reversible" Turing machines.

\begin{proposition}[{\cite[Theorem 1]{lecerf1963machines}}]
    \AP\label{prop:halting-problem-detrevTM}
    The halting problem on deterministic "reversible" Turing machines is undecidable.
\end{proposition}

For more details and pointers on "reversible" Turing machines, see \cite[Chapter 5]{Morita2017}.

\begin{restatable}{lemma}{reachableregularity}
    \AP\labelwithproof{lem:reachable-regularity}
    The "reachable regularity problem" is undecidable.
\end{restatable}

\begin{proof}[Proof sketch]
    By reducing the halting problem on deterministic "reversible" Turing machines,
    in such a way that the "reachable configurations" whose
    state $q$ coincide with the state of the original machine are
    of the form $(u q v a^n b^n)$ where $(u q v)$ is a configuration of the original machine,
    $a$ and $b$ are new symbols,
    and $n\in\N$. Transitions are defined in such a way that the new machine is a
    "wf-RTM": this is implemented by having, for every transition $uqv \to u'q'v'$ of the original machine and every $n\in \N$, a (multi-step) transition $(u q v a^n b^n) \to^* (u' q' v' a^{n+1} b^{n+1})$---and is illustrated in \Cref{fig:reachable-regularity}.
	\begin{figure}[htb]
		\centering
		\begin{tikzpicture}[>={Classical TikZ Rightarrow}, font = \small]
	\node (0) at (0,0) {0};
	\node[right = 0cm of 0] (1) {0};
	\node[right = 0cm of 1] (2) {1};
	\node[right = 0cm of 2] (3) {0};
	\node[right = 0cm of 3] (4) {1};
	\node[right = 0cm of 4] (5) {$a\vphantom{b}$};
	\node[right = 0cm of 5] (6) {$a\vphantom{b}$};
	\node[right = 0cm of 6] (7) {$a\vphantom{b}$};
	\node[right = 0cm of 7] (8) {$b$};
	\node[right = 0cm of 8] (9) {$b$};
	\node[right = 0cm of 9] (10) {$b$};

	\draw[rounded corners=4pt] (0.south west) rectangle (10.north east);
	\draw[->, thick] ($(5.north)+(0,.4)$) -- ($(5.north)+(0,.1)$);
	\node[above=.05cm, circle, draw=cBlue, fill=cBlue, opacity=.5, text opacity=1, inner sep=1.5pt] at ($(5.north)+(0,.4)$) {$p$};

	\node[below = 1.2cm of 0] (0') {0};
	\node[right = 0cm of 0'] (1') {0};
	\node[right = 0cm of 1'] (2') {1};
	\node[right = 0cm of 2'] (3') {0};
	\node[right = 0cm of 3'] (4') {1};
	\node[right = 0cm of 4'] (5') {1};
	\node[right = 0cm of 5'] (6') {$a\vphantom{b}$};
	\node[right = 0cm of 6'] (7') {$a\vphantom{b}$};
	\node[right = 0cm of 7'] (8') {$b$};
	\node[right = 0cm of 8'] (9') {$b$};
	\node[right = 0cm of 9'] (10') {$b$};

	\draw[rounded corners=4pt] (0'.south west) rectangle (10'.north east);
	\draw[->, thick] ($(6'.north)+(0,.4)$) -- ($(6'.north)+(0,.1)$);
	\node[above=.05cm, circle, draw=cRed, fill=cRed, opacity=.5, text opacity=1, inner sep=1.5pt] at ($(6'.north)+(0,.4)$) {$\phantom{q}$};

	\node[below = 1.2cm of 0'] (0'') {0};
	\node[right = 0cm of 0''] (1'') {0};
	\node[right = 0cm of 1''] (2'') {1};
	\node[right = 0cm of 2''] (3'') {0};
	\node[right = 0cm of 3''] (4'') {1};
	\node[right = 0cm of 4''] (5'') {1};
	\node[right = 0cm of 5''] (6'') {$a\vphantom{b}$};
	\node[right = 0cm of 6''] (7'') {$a\vphantom{b}$};
	\node[right = 0cm of 7''] (8'') {$a\vphantom{b}$};
	\node[right = 0cm of 8''] (9'') {$a\vphantom{b}$};
	\node[right = 0cm of 9''] (10'') {$b$};

	\draw[rounded corners=4pt] (0''.south west) rectangle (10''.north east);
	\draw[->, thick] ($(10''.north)+(0,.4)$) -- ($(10''.north)+(0,.1)$);
	\node[above=.05cm, circle, draw=cRed, fill=cRed, opacity=.5, text opacity=1, inner sep=1.5pt] at ($(10''.north)+(0,.4)$) {$\phantom{q}$};

	\node[below = 1.2cm of 0''] (0''') {0};
	\node[right = 0cm of 0'''] (1''') {0};
	\node[right = 0cm of 1'''] (2''') {1};
	\node[right = 0cm of 2'''] (3''') {0};
	\node[right = 0cm of 3'''] (4''') {1};
	\node[right = 0cm of 4'''] (5''') {1};
	\node[right = 0cm of 5'''] (6''') {$a\vphantom{b}$};
	\node[right = 0cm of 6'''] (7''') {$a\vphantom{b}$};
	\node[right = 0cm of 7'''] (8''') {$a\vphantom{b}$};
	\node[right = 0cm of 8'''] (9''') {$a\vphantom{b}$};
	\node[right = 0cm of 9'''] (10''') {$b$};
	\node[right = 0cm of 10'''] (11''') {$b$};
	\node[right = 0cm of 11'''] (12''') {$b$};
	\node[right = 0cm of 12'''] (13''') {$b$};

	\draw[rounded corners=4pt] (0'''.south west) rectangle (13'''.north east);
	\draw[->, thick] ($(13'''.north)+(0,.4)$) -- ($(13'''.north)+(0,.1)$);
	\node[above=.05cm, circle, draw=cRed, fill=cRed, opacity=.5, text opacity=1, inner sep=1.5pt] at ($(13'''.north)+(0,.4)$) {$\phantom{q}$};

	\node[below = 1.2cm of 0'''] (0'''') {0};
	\node[right = 0cm of 0''''] (1'''') {0};
	\node[right = 0cm of 1''''] (2'''') {1};
	\node[right = 0cm of 2''''] (3'''') {0};
	\node[right = 0cm of 3''''] (4'''') {1};
	\node[right = 0cm of 4''''] (5'''') {1};
	\node[right = 0cm of 5''''] (6'''') {$a\vphantom{b}$};
	\node[right = 0cm of 6''''] (7'''') {$a\vphantom{b}$};
	\node[right = 0cm of 7''''] (8'''') {$a\vphantom{b}$};
	\node[right = 0cm of 8''''] (9'''') {$a\vphantom{b}$};
	\node[right = 0cm of 9''''] (10'''') {$b$};
	\node[right = 0cm of 10''''] (11'''') {$b$};
	\node[right = 0cm of 11''''] (12'''') {$b$};
	\node[right = 0cm of 12''''] (13'''') {$b$};

	\draw[rounded corners=4pt] (0''''.south west) rectangle (13''''.north east);
	\draw[->, thick] ($(6''''.north)+(0,.4)$) -- ($(6''''.north)+(0,.1)$);
	\node[above=.05cm, circle, draw=cBlue, fill=cBlue, opacity=.5, text opacity=1, inner sep=1.5pt] at ($(6''''.north)+(0,.4)$) {$q$};

	\draw[->, dashed] ($(10.east)+(.3,-.2)$) to[bend left=50]
		node[midway, right=.15cm, align=left, text width=3.5cm, font=\footnotesize] {simulate $T$}
		($(10'.east)+(.3,.2)$);

	\draw[->, dashed] ($(10'.east)+(.3,-.2)$) to[bend left=50]
		node[midway, right=.15cm, align=left, text width=3.5cm, font=\footnotesize]
			{overwrite the first two $b$'s with $a$'s}
		($(10''.east)+(.3,.2)$);

	\draw[->, dashed] ($(10''.east)+(1.5,-.2)$) to[bend left=50]
		node[midway, right=.15cm, align=left, text width=3.5cm, font=\footnotesize]
			{append three $b$'s}
		($(13'''.east)+(.3,.2)$);

	\draw[->, dashed] ($(13'''.east)+(.3,-.2)$) to[bend left=50]
		node[midway, right=.15cm, align=left, text width=3.5cm, font=\footnotesize]
			{go back to the new position, in the new state}
		($(13''''.east)+(.3,.2)$);
\end{tikzpicture}
		\caption{
			\AP\label{fig:reachable-regularity}
			Encoding of a single transition of the form
			``when reading a blank in state $\color{cBlue} p$, write a
			$1$, go in state $\color{cBlue} q$ and move right''
			of the machine $T$ in the machine $T'$
			in the proof of \Cref{lem:reachable-regularity}.
			Red unlabelled states represent states of $T'$
			that are not originally present in $T$.
		}
	\end{figure}
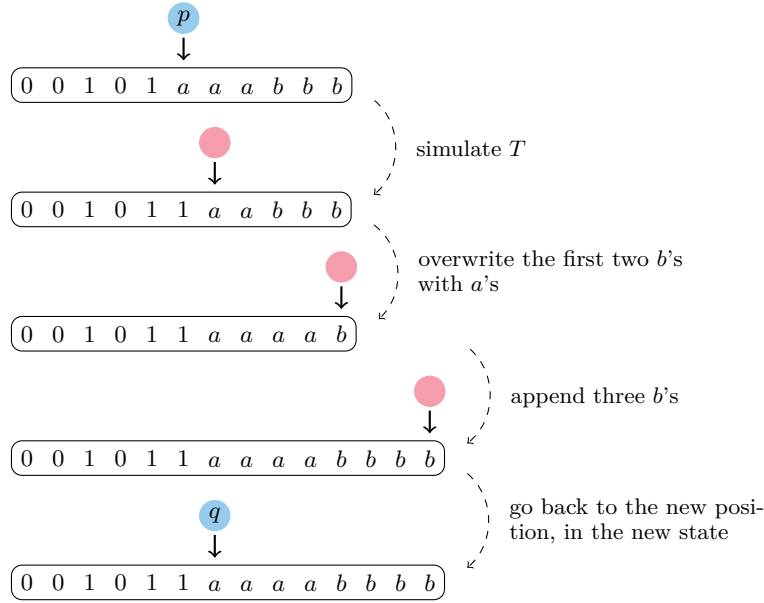
	Moreover:
    \begin{itemize}
        \item if the original machine was halting, then the "reachable configurations"
            of the new one are finite and hence regular;
        \item otherwise, the set of "reachable configurations" is not regular,
            which follows from the non-regularity of any infinite subset of $\{a^n b^n \mid n \in \N\}$.
    \end{itemize}
    
    See
    \ifarxiv
     \Cref{proof-lem:reachable-regularity}
    \else
     the "full version"
    \fi
    for more details.
\end{proof}

\subsection[Undecidability of the $k$-Regular Colorability Problem]{Undecidability of the $\boldsymbol{k}$-Regular Colorability Problem}
We can now show undecidability for the "$k$-regular colorability problem" by reduction from the "reachable regularity problem" as defined before.

\begin{fact}
    \AP\label{fact:initial-nodes-are-regular}
    Given an "automatic graph", the set of nodes with no predecessor is effectively a regular language. 
\end{fact}

\kregcolundec

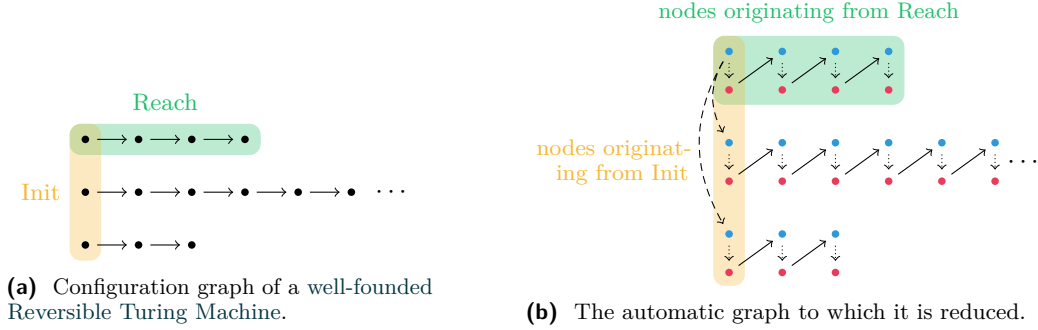
\begin{figure}
    \centering
    \begin{subfigure}[b]{.4\textwidth}
        \centering
        \begin{tikzpicture}[>={Classical TikZ Rightarrow},
	node distance=.7cm,
	every node/.append style={fill, circle, inner sep=1pt, outer sep=3pt}]
	
	\fill[rounded corners, draw=cGreen, fill=cGreen, opacity=.3]
		(-.2,-.2) rectangle (2.3, .2);
	\fill[rounded corners, draw=cYellow, fill=cYellow, opacity=.3]
		(-.2,.2) rectangle (.2, -1.6);

	\node (a0) at (0,0) {};
	\foreach \i in {0,...,2} {
		\pgfmathtruncatemacro{\next}{\i + 1}
		\node [right of = a\i] (a\next) {};
		\draw[->] (a\i) edge (a\next);
	}
	
	\node [below of = a0] (b0) {};
	\foreach \i in {0,...,4} {
		\pgfmathtruncatemacro{\next}{\i + 1}
		\node [right of = b\i] (b\next) {};
		\draw[->] (b\i) edge (b\next);
	}
	\node[draw=none, fill=none, right = 0cm of b5] (binf) {$\cdots$};

	\node [below of = b0] (c0) {};
	\foreach \i in {0,1} {
		\pgfmathtruncatemacro{\next}{\i + 1}
		\node [right of = c\i] (c\next) {};
		\draw[->] (c\i) edge (c\next);
	}

	\node[draw=none, fill=none, font=\small] [left = 0cm of b0] {$\InitC{cYellow}$};
	\node[draw=none, fill=none, font=\small, align=center] [above = -.1cm of $(a1)!0.5!(a2)$]
		{$\ReachC{cGreen}$};
\end{tikzpicture}
        \caption{
            \AP\label{subfig:config-graph-wf-RTM}
            Configuration graph of a "well-founded Reversible Turing Machine".
        }
    \end{subfigure}
    \hspace{1cm}
    \begin{subfigure}[b]{.5\textwidth}
        \centering
        \begin{tikzpicture}[>={Classical TikZ Rightarrow},
	node distance=.7cm,
	every node/.append style={fill, circle, inner sep=1pt, outer sep=3pt}]

	\fill[rounded corners, draw=cGreen, fill=cGreen, opacity=.3]
		(-.2,-.7) rectangle (2.3, .2);
	\fill[rounded corners, draw=cYellow, fill=cYellow, opacity=.3]
		(-.2,.2) rectangle (.2, -3.1);

	\node[cBlue] (a0) at (0,0) {};
	\node[cRed] (a'0) [below = .2cm of a0] {};
	\draw[->, densely dotted] (a0) edge (a'0);
	\foreach \i in {0,...,2} {
		\pgfmathtruncatemacro{\next}{\i + 1}
		\node[cBlue, right of = a\i] (a\next) {};
		\node[cRed, right of = a'\i] (a'\next) {};
		\draw[->] (a'\i) edge (a\next);
		\draw[->, densely dotted] (a\next) edge (a'\next);
	}
	
	\node[cBlue, below of = a'0] (b0) {};
	\node[cRed] (b'0) [below = .2cm of b0] {};
	\draw[->, densely dotted] (b0) edge (b'0);
	\foreach \i in {0,...,4} {
		\pgfmathtruncatemacro{\next}{\i + 1}
		\node[cBlue, right of = b\i] (b\next) {};
		\node[cRed, right of = b'\i] (b'\next) {};
		\draw[->] (b'\i) edge (b\next);
		\draw[->, densely dotted] (b\next) edge (b'\next);
	}
	\node[draw=none, fill=none, right = 0cm of $(b5)!0.5!(b'5)$] (binf) {$\cdots$};

	\node[cBlue, below of = b'0] (c0) {};
	\node[cRed] (c'0) [below = .2cm of c0] {};
	\draw[->, densely dotted] (c0) edge (c'0);
	\foreach \i in {0,1} {
		\pgfmathtruncatemacro{\next}{\i + 1}
		\node[cBlue, right of = c\i] (c\next) {};
		\node[cRed, right of = c'\i] (c'\next) {};
		\draw[->] (c'\i) edge (c\next);
		\draw[->, densely dotted] (c\next) edge (c'\next);
	}

	\draw[->, densely dashed] (a0) edge[bend right] (b0);
	\draw[->, densely dashed] (a0) edge[bend right] (c0);

	\node[draw=none, fill=none, cYellow, font=\footnotesize, text width=2cm, align=right]
		[left = .35cm of $(b0)!0.5!(b'0)$]
		{nodes originating from	$\InitC{cYellow}$};
	\node[draw=none, fill=none, cGreen, font=\footnotesize, align=center]
		[above = -1.6cm of $(a1)!0.5!(a2)$]
		{nodes originating from $\ReachC{cGreen}$};
\end{tikzpicture}
        \caption{
            \AP\label{subfig:reduction-wf-RTM}
            The automatic graph to which it is reduced.
        }
    \end{subfigure}
    \caption{
        \AP\label{fig:reduction-wf-RTM-to-coloring}
        Reduction used in the proof of \Cref{thm:k-reg-col-undec}.
    }
\end{figure}
\begin{proof}
	\proofcase{Lower bound.}
    By reduction from the "reachable regularity problem" for "wf-RTM"s
    (\Cref{lem:reachable-regularity}). We first show it for $k=2$.
    \AP Given a "wf-RTM" $T$, let $c_{\textit{init}}$ be its "initial configuration".
    Observe that the set $\intro*\Init$ of all vertices of $\confGraph$ with in-degree $0$ is an effective regular language (by \Cref{fact:initial-nodes-are-regular}), and that $c_{\textit{init}} \in \Init$. Let $B$ and $R$ be fresh symbols. 
    Consider the "automatic graph" $\AutGraph{L}{E}$ for $L = \set{B,R} \times \configs$, having 
    an edge from $(z,c) \in \{B,R\} \times \configs$ to $(z',c') \in \{B,R\} \times \configs$ if either 
    \begin{enumerate}
        \item $(z,z') = (B,R)$ and $c=c'$;
        \item $(z,z') = (R,B)$ and there is an edge from $c$ to $c'$ in $\confGraph$; or
        \item $(z,z') = (B,B)$, $c = c_{\textit{init}}$ and $c' \in \Init \setminus \set{c_{\textit{init}}}$.
    \end{enumerate}
Fresh symbols $B$ and $R$ are utilized to represent two versions of each configuration - one in Blue and one in Red. This graph is depicted
    on \Cref{fig:reduction-wf-RTM-to-coloring}.
    Note that $\AutGraph{L}{E}$ is connected and "2-colorable": in fact, it is a directed (possibly infinite) tree with root $(B,c_{\textit{init}})$. 
    
    We claim that $\AutGraph{L}{E}$ is "$2$-regular colorable" if, and only if, the set of "reachable configurations" of $T$ is a regular language. 
    In fact, up to permuting the two-colors, 
  $\AutGraph{L}{E}$ admits a unique 2-coloring, defined by:
    \[
        C_1 ~~ \defeq ~~ \{B\} \times \Reach ~~\cup~~ \{R\} \times (\configs \setminus \Reach)
    \]
    and $C_2$ is the complement of $C_1$.
    If $\Reach$ is regular, then so is $C_1$. Dually, if $C_1$ is regular, then
    $\Reach$ is the set of configurations $c$ such that $(B,c) \in C_1$ and hence is regular.
    It follows that $\AutGraph{\A^*}{E}$ is "$2$-regular colorable" if and only if
    the "reachable configurations" of $T$ are regular, which concludes the proof for $k=2$.



    To prove the statement for any $k>2$, we define $\AutGraph{L}{E_k}$ as the result of adding a $(k-2)$-clique to $\AutGraph{L}{E}$ and adding an edge from every vertex of the clique to every vertex incident to an edge of $E$. This forces the clique to use $k-2$ colors that cannot be used in the remaining part of the graph and the proof is then analogous.

	\proofcase{Upper-bound.} We show that the problem is recursively enumerable. Let us define a $k$-colored automaton like a regular (complete) DFA, except that instead of having
	a set of final states, it has a partition $\langle C_1,\hdots,C_k \rangle$ of its states.
	Such an automaton recognizes a regular coloring $\A^* \to \set{1, \dotsc, k}$.
	Given an "automatic graph" $\AutGraph{L}{R}$---specified by
	 NFA's $\+A_1$ and $\+A_2$ recognizing $L$ and $\lotimes R$  respectively--- and a $k$-colored automaton $\+B$,
	we can build, by a product construction, an NFA $\+A'_2$  which accepts
	all $u \otimes v \in \lotimes R$ such that the color of $u$ is distinct from the color of $v$.
	Then, $\+A'_2$ is equivalent to $\+A_2$ if, and only if, $\+B$ describes a proper "$k$-coloring" 
	of $\AutGraph{L}{R}$. The "RE" upper-bound of the "$k$-regular colorability problem" follows: it 
	suffices to enumerate all $k$-colored automata and check for equivalence.
\end{proof}

Note that this reduction provides an easy way of building
graphs in the shape of \Cref{subfig:reduction-wf-RTM} that are "2-colorable" (in fact, they are trees) but not "2-regular colorable". In fact, we can provide a slightly more
direct construction.

\begin{example}
    \AP\label{ex:tree-not-2-reg-colorable}
    On the alphabet $\A = \{a,b\}$, the tree $\+T$ depicted in \Cref{fig:tree-not-2reg-color} whose set of vertices is $V = a^*b^*$ and whose set 
    of edges is $E = E_{\mathrm{incr}} \cup E_{\mathrm{init}}$, with 
    \begin{align*}
        E_{\mathrm{incr}} & = \{(a^pb^q,\, a^{p+1}b^{q+1}) \mid p,q \in \N\} \\
        E_{\mathrm{init}} & = \{(\varepsilon,\, a^p) \mid p \in \N\} \cup \{(\varepsilon,\, b^q) \mid q \in \N\}, 
    \end{align*}    
    is "automatic" but not "2-regular colorable". 
    Indeed, its only "2-coloring"
    consists in partitioning the vertices of $\+T$ into
    \[
        C = \{a^n b^n \mid n \in 2\N\}
            \cup \{a^p b^q \mid p > q \text{ and $q$ is odd}\}
            \cup \{a^p b^q \mid p < q \text{ and $p$ is odd}\}
    \]
    and its complement $V \setminus C$.
    Let $P = \{a^p b^q \mid p, q \in 2\N\} = (aa)^*(bb)^*$:
    $P$ is regular, yet $C \cap P = \{a^n b^n \mid n \in 2\N\}$ is not.
    Hence, $C$ is not regular, and thus $\+T$ is not "2-regular colorable".
    \qed 
\end{example}

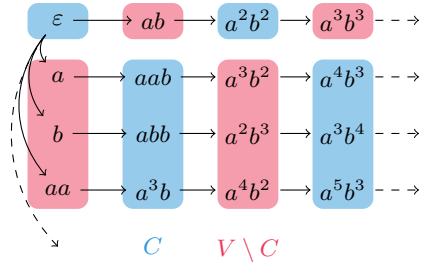
\begin{figure}[htb]
    \centering
    \begin{tikzpicture}[>={Classical TikZ Rightarrow},
	font=\small,
	node distance=.5cm]

	\fill[rounded corners, fill=cBlue, opacity=.5]
		(-.4,-0.25) rectangle (.4,0.23)
		(2.1,-0.25) rectangle (2.9,0.23)
		(0.85,-0.52) rectangle (1.65, -2.5)
		(3.35,-0.52) rectangle (4.15, -2.5);
	
	\fill[rounded corners, fill=cRed, opacity=.5, xshift=1.25cm]
		(-.4,-0.25) rectangle (.4,0.23)
		(2.1,-0.25) rectangle (2.9,0.23);
	\fill[rounded corners, fill=cRed, opacity=.5, xshift=-1.25cm]
		(0.85,-0.52) rectangle (1.65, -2.5)
		(3.35,-0.52) rectangle (4.15, -2.5);

	\node (eps) at (0,0) {$\varepsilon$};
	\node (ab) at (1.25,0) {$ab$};
	\node (aabb) at (2.5,0) {$a^2b^2$};
	\node (aaabbb) at (3.75,0) {$a^3b^3$};
	
	\node (a) at (0,-.75) {$a$};
	\node (aab) at (1.25,-.75) {$aab$};
	\node (aaabb) at (2.5,-.75) {$a^3b^2$};
	\node (aaaabbb) at (3.75,-.75) {$a^4b^3$};
	
	\node (b) at (0,-1.5) {$b$};
	\node (abb) at (1.25,-1.5) {$abb$};
	\node (aabbb) at (2.5,-1.5) {$a^2b^3$};
	\node (aaabbbb) at (3.75,-1.5) {$a^3b^4$};

	\node (aa) at (0,-2.25) {$aa$};
	\node (aaab) at (1.25,-2.25) {$a^3b$};
	\node (aaaabb) at (2.5,-2.25) {$a^4b^2$};
	\node (aaaaabbb) at (3.75,-2.25) {$a^5b^3$};

	\draw[->] (eps) to (ab);
	\draw[->] (ab) to (aabb);
	\draw[->] (aabb) to (aaabbb);
	\draw[->, dashed] (aaabbb) to ($(aaabbb)+(1,0)$);

	\draw[->] (a) to (aab);
	\draw[->] (aab) to (aaabb);
	\draw[->] (aaabb) to (aaaabbb);
	\draw[->, dashed] (aaaabbb) to ($(aaaabbb)+(1,0)$);

	\draw[->] (b) to (abb);
	\draw[->] (abb) to (aabbb);
	\draw[->] (aabbb) to (aaabbbb);
	\draw[->, dashed] (aaabbbb) to ($(aaabbbb)+(1,0)$);

	\draw[->] (aa) to (aaab);
	\draw[->] (aaab) to (aaaabb);
	\draw[->] (aaaabb) to (aaaaabbb);
	\draw[->, dashed] (aaaaabbb) to ($(aaaaabbb)+(1,0)$);

	\draw[->] (eps) edge[bend right=40] (a)
		edge[bend right=40] (b)
		edge[bend right=40] (aa)
		edge[dashed, bend right=40] ($(aa)+(0,-.75)$);

	\node[below = .25cm of aaab, color=cBlue] {$C$}; 
	\node[below = .25cm of aaaabb, color=cRed] {$V \setminus C$}; 
\end{tikzpicture}
    \caption{
        \label{fig:tree-not-2reg-color}
        The "automatic tree@automatic graph" $\+T$ of \Cref{ex:tree-not-2-reg-colorable},
        and its unique "2-coloring" $(C, V\setminus C)$, which is not "regular@@coloring".
    }
\end{figure}
\section{Separability for Bounded Recognizable Relations}

In this section we capitalize on the undecidability result of the previous section, showing how this implies the undecidability for the "separability problem" on two natural classes of bounded "recognizable relations", namely: $\kREC$, and $\kPROD$.
Remember that, for any $k$, $\reintro*\kPROD$ is the subclass of $\REC$ consisting of unions of $k$ cross-products of regular languages (which is a subclass of $\kREC[2^{2k}]$).

\subparagraph*{$\kREC$-separability.} 
First, observe that the $\kREC[1]$-"separability problem" is trivially decidable, since the only possible "separator" is $\A^* \times \A^*$. However, for any other $k>1$, the problem is undecidable.

\begin{proposition}
    The $\kREC$-"separability problem" is undecidable, for every $k>1$.
\end{proposition}
\begin{proof}
A consequence of the reduction from the "$k$-regular colorability problem" of \Cref{thm:reg-colorability-equiv-separability}, combined with the undecidability of the latter for every $k>1$ (\Cref{thm:k-reg-col-undec}).
\end{proof}

\subparagraph*{$\kPROD$-separability.} 
On the $\kPROD$ hierarchy we will find the same phenomenon. In particular the case $k=1$ is also trivially decidable.

\begin{proposition}
    The $\kPROD[1]$-"separability problem" is decidable.
\end{proposition}
\begin{proof}
    Given two "automatic relations" $R_1, R_2$, there exists $S \in $ \kPROD[1]
    that "separates" $R_1$ from $R_2$ if and only if $\pi_1(R_1)\times \pi_2(R_1)$
    "separates" $R_1$ from $R_2$.
\end{proof}

As soon as $k>1$, the $\kPROD$ "separability problem" becomes undecidable. This is a consequence of the following simple lemma.

\begin{lemma}
    A symmetric "automatic relation" $R$ and the identity $\Id$ are "separable" by a relation in $\kPROD[2]$ iff they have a "separator" of the form $(A \times B) \cup (B \times A)$.
\end{lemma}
\begin{proof}
    Assume that $S \in \kPROD[2]$ "separates" $R$ from $\Id$.
    Then $R \subseteq S$, but since $R$ is symmetric, $R = R^{-1} \subseteq S^{-1}$ so
    $R \subseteq S \cap S^{-1}$, and hence $R \subseteq S \cap S^{-1}$.
    Moreover, since $S$ has a trivial intersection with $\Id$, so does $S \cap S^{-1}$.
    Hence, $S \cap S^{-1}$ "separates" $R$ from $\Id$.

    Since $S \in \kPROD[2]$, there exists $A_1,A_2,B_1,B_2 \subseteq \A^*$ such that
    $S = A_1 \times B_1 \cup B_2 \times A_2$.
    Note that $S \cap \Id = \emptyset$ yields $A_i \cap B_i = \emptyset$ for each $i \in \{1,2\}$.
    Finally:
    \begin{align*}
        S \cap S^{-1} &
        =
            \bigl( A_1 \times B_1 \cup B_2 \times A_2 \bigr)
            \cap \bigl( B_1 \times A_1 \cup A_2 \times B_2 \bigr) \\
        &
        =
            \bigl( (A_1 \times B_1) \cap (B_1 \times A_1) \bigr)
            \cup \bigl( (A_1 \times B_1) \cap (A_2 \times B_2) \bigr) \\
        &
        \hphantom{=\;} \cup \bigl( (B_2 \times A_2) \cap (B_1 \times A_1) \bigr)
            \cup \bigl( (B_2 \times A_2) \cap (A_2 \times B_2) \bigr) \\
        &
        =
            \bigl( \overbrace{(A_1 \cap B_1) \times (A_1 \cap B_1)}^{= \emptyset} \bigr)
            \cup \bigl( (A_1 \cap A_2) \times (B_1 \cap B_2) \bigr) \\
        &
        \hphantom{=\;} \cup \bigl( (B_1 \cap B_2) \times (A_1 \cap A_2) \bigr)
            \cup \bigl( \underbrace{(A_2 \cap B_2) \times (A_2 \cap B_2)}_{= \emptyset} \bigr)\\
        &
        = \bigl( (A_1 \cap A_2) \times (B_1 \cap B_2) \bigr)
            \cup \bigl( (B_1 \cap B_2) \times (A_1 \cap A_2) \bigr).\qedhere
    \end{align*}
\end{proof}

We can then establish the following: 

\begin{corollary}\AP\label{cor:2reg-2prod}
    A symmetric "automatic relation" $R$ and $\Id$ are "separable" by a relation in $\kPROD[2]$ if{f} $\AutGraph{\A^*}{R}$ is "$2$-regular colorable".
\end{corollary}
\begin{proof}
    By observing that for any symmetric relation $R \subseteq \A^* \times \A^*$, we have that $A,B \subseteq \A^*$ is a coloring of $\AutGraph{\A^*}{R}$ if, and only if, $(A \times B) \cup (B \times A)$ "separates" $R$ from $\Id$.
\end{proof}

We can now easily show undecidability for the $\kPROD[2]$ "separability problem" by reduction from the "$2$-regular colorability problem".
\begin{lemma}\AP\label{lem:aut-2prod-sep-undec}
    The $\kPROD[2]$-"separability problem" is undecidable.
\end{lemma}
\begin{proof}
    By reduction from the "$2$-regular colorability problem" on "automatic graphs", which is undecidable by \Cref{thm:k-reg-col-undec}. Let $\AutGraph{L}{R}$ be an "automatic graph" and $\AutGraph{L}{R'}$ the symmetric closure of $\AutGraph{L}{R}$. It follows that $\AutGraph{L}{R'}$ is still "automatic" and that there is a "$2$-regular coloring" for $\AutGraph{L}{R'}$ if{f} there is a "$2$-regular coloring" for $\AutGraph{L}{R}$ (the same coloring in fact).
    Thus, by \Cref{cor:2reg-2prod}, $\AutGraph{L}{R}$ is "$2$-regular colorable" if{f} 
    there is a $\kPROD[2]$ relation that "separates" $R'$ from $\Id$.
\end{proof}

Further, this implies undecidability for every larger $k$:
\kprodundecidable
\begin{figure}[htb]
    \centering
    \scalebox{1.2}{
        \begin{tikzpicture}[>=stealth, use Hobby shortcut]
	\newcommand{\curveA}{(0,0) .. (.5,-.3) .. (2,2) .. (1.1,.8)}
	\newcommand{\curveB}{(.65,.9) .. (.55, 2) .. (-.1, 1.6) .. (-.9, 1.4) .. (-1.3, 1) .. (-1.1, .5) .. (0, 1) .. (.45, .9)};

	\draw[rounded corners=2pt, draw=cGrey, fill=cGrey, opacity=.5] (1.14,.85) rectangle (-.15,2.1);
	\draw[rounded corners=2pt, draw=cGrey, fill=cGrey, opacity=.5] (-.15,1.7) rectangle (-1.4,.4);

	\draw[
		closed,
		fill=cRed,
		draw=cRed,
		opacity=.5
	] \curveA;

	\draw[
		closed,
		fill=cBlue,
		draw=cBlue,
		opacity=.5
	] \curveB;

	\node[cRed, font=\small] at (1.5, -.6) {$R_2$};
	\node[cBlue, font=\small] at (-.7, 2) {$R_1$};
	\node[cGrey, font=\small] at (-1.65, 1.05) {$S$};
	
	\foreach \x in {0,...,2} {
		\coordinate (ca\x) at ($(3.75,1.25)+(.75*\x,0)$);
		\coordinate (cb\x) at ($(3.75,.5)+(.75*\x,0)$);
	}

	\draw[rounded corners=2pt, draw=cGrey, fill=cGrey, opacity=.5]
		($(ca0)+(0,-.1)$) rectangle ($(cb0)+(-.15,.1)$)
		($(ca1)+(0,-.1)$) rectangle ($(cb1)+(-.15,.1)$)
		($(ca2)+(0,-.1)$) rectangle ($(cb2)+(-.15,.1)$);
	\node[cGrey, font=\small, right] at (5.4, .875) {$S'\setminus S$};

	\foreach \x in {0,...,2} {
		\node[circle, fill, inner sep=1pt, outer sep=2pt] (a\x) at (ca\x) {};
			\node[above = 0cm of a\x, font=\small] {$a_{\x}$};
		\node[circle, fill, inner sep=1pt, outer sep=2pt] (b\x) at (cb\x) {};
			\node[below = 0cm of b\x, font=\small] {$b_{\x}$};
	}

	\foreach \x in {1,2} {
		\draw[<->, cRed] (a0) edge (b\x);
	}
	\foreach \x in {0,2} {
		\draw[<->, cRed] (a1) edge (b\x);
	}
	\foreach \x in {0,1} {
		\draw[<->, cRed] (a2) edge (b\x);
	}
	\foreach \x in {0,...,2} {
		\draw[->, cBlue, bend right=20] (a\x) edge (b\x);
		\draw[->, cRed, bend right=20] (b\x) edge (a\x);
	}

	\draw[<-, cRed] (1.44, 2.46) .. (1.64, 2.56) .. (2, 2.55) ..  ($(a0)+(-.2,0)$) .. ($(a0)+(-.08,0)$);
	\foreach \x in {1,2} {
		\draw[-, cRed] (1.45, 2.465) .. (1.64, 2.56) .. (2, 2.55) ..  ($(a\x)+(-.2,0)$) .. ($(a\x)+(-.08,0)$);
	}
	\node[circle, draw, inner sep=1pt, outer sep=1pt, fill=white] at (1.37, 2.45) {};
	\node[below=-2pt, font=\tiny, align=center] at (1.37, 2.45) {$u \in \A^*$};

	\foreach \x in {0,...,2} {
		\draw[<-, cRed] ($(b\x)+(-.08,-.05)$) .. ($(b\x)+(-.2,-.25)$) .. ($(b\x)+(-.3,-.7)$) .. (3.2, -.5) ..  (2.85, -.55) .. (2.65, -.55);
	}
	\node[circle, draw, inner sep=1pt, outer sep=1pt, fill=white] at (2.6, -.6) {};
	\node[below=-2pt, font=\tiny, align=center] at (2.6, -.6) {$v \in \A^*$};

\end{tikzpicture}
    }
    \caption{
        \AP\label{fig:2prod-to-kprod}
        Construction in the proof of \Cref{thm:kprod-undecidable} for $k = 5$. $S$ is depicted as the union of two (gray) rectangles since $S \in \kPROD[2]$.
        The relation $R'_1$ is obtained from $R_1$ (blue shape) by adding all blue edges,
        namely $(a_i, b_i)$ for $1\leq i \leq k-2$. The relation $R'_2$ is obtained from $R_2$ (red shape) by adding
        all red edges, namely every other edge involving a vertex $a_i$ or $b_i$.
        Finally, $S'$ (five gray rectangles) is obtained from $S$ by adding
        each $\{a_i\} \times \{b_i\}$.
    }
\end{figure}
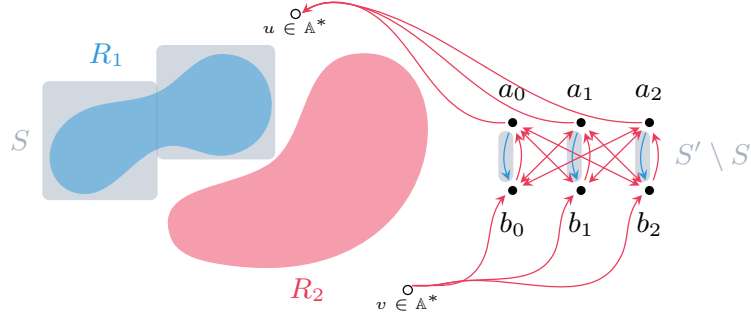

\begin{proof}
    The case $k=2$ is shown in \Cref{lem:aut-2prod-sep-undec}, so suppose $k>2$.
    The proof goes by reduction from the $\kPROD[2]$-"separability problem". Let $R_1,R_2$ be a pair of "automatic relations" over an alphabet $\A$. Consider the alphabet extended with $2(k-2)$ fresh symbols $\A' = \A \dcup \set{a_1, \dotsc, a_{k-2}, b_1, \dotsc, b_{k-2}}$. We build "automatic relations" $R'_1,R'_2$ over $\A'$ such that $(R_1,R_2)$ are $\kPROD[2]$ "separable" over $\A$ if{f} $(R'_1,R'_2)$ are $\kPROD$ "separable" over $\A'$.

    Let $R'_1 = R_1 \dcup \set{ (a_i,b_i) : 1 \leq i \leq k-2}$ and 
    \begin{align*}
    R'_2 =  R_2 \dcup {}&\set{(a_i,w) : w \in \A^*, 1 \leq i \leq k-2} \dcup {}\\
                    &\set{(w,b_i) : w \in \A^*, 1 \leq i \leq k-2} \dcup {}\\
                    &\set{(a_i,b_j) : 1 \leq i,j \leq k-2, i \neq j} \dcup {}\\
                    &\set{(b_i,a_j) : 1 \leq i,j \leq k-2}
    \end{align*}
    
    If $(R_1,R_2)$ has a $\kPROD[2]$ "separator" $S$, then $\tilde S \dcup \set{(a_i,b_i) : 1 \leq i \leq k-2}$ is a $\kPROD$ "separator" of $(R'_1,R'_2)$.

    Conversely, if $S' = (A_1 \times B_1) \cup \dotsb \cup (A_k \times B_k)$ is a $\kPROD$ "separator" of $(R'_1,R'_2)$, then for every $i$ there must be some $j_i$ such that $A_{j_i} \times B_{j_i}$ contains $(a_i,b_i)$. Observe that 
    \begin{itemize}
        \item $A_{j_i} \cup B_{j_i}$ cannot contain any $a_{i'}$ or $b_{i'}$ for $i' \neq i$, and
        \item $A_{j_i} \cup B_{j_i}$ cannot contain any $w \in \A^*$;
    \end{itemize}
    since otherwise we would have $(A_{j_i} \times B_{j_i}) \cap R'_2 \neq \emptyset$.
    Hence, $\set{i \mapsto j_i}_i$ is injective, and thus $S'$ is of the form $S' = (A_1 \times B_1) \cup (A_2 \times B_2)  \cup (\set{a_1} \times \set{b_1}) \cup \dotsb \cup (\set{a_{k-2}} \times \set{b_{k-2}})$. We can further assume that $A_1,B_1,A_2,B_2$ do not contain any $a_i$ or $b_i$ since otherwise we can remove them preserving the property of being a $\kPROD$ "separator" of $R'_1$ and $R'_2$.
    Hence, $S \defeq (A_1 \times B_1) \cup (A_2 \times B_2)$ must cover $R_1$ and be disjoint from $R_2$, obtaining that $S$ is a $\kPROD[2]$ "separator" of $R_1$ and $R_2$.
\end{proof}


\section{Definability for Bounded Recognizable Relations}
Up until now, we have examined two hierarchies of bounded recognizable relations, namely $\kPROD$ and $\kREC$. 
Our previous analysis demonstrated that, for any element in these hierarchies (where $k>1$), the "separability problem" is undecidable. Nevertheless, 
we will now establish that the "definability problem" is decidable.

\knowledgenewrobustcmd{\autequiv}[1][R]{\cmdkl{\sim_{#1}}}%
\AP
Given an "automatic" relation $R \subseteq \A^* \times \A^*$, consider the "automatic" equivalence relation $\intro*\autequiv \subseteq \A^* \times \A^*$, defined as $w \mathrel{\autequiv} w'$ if for every $v \in \A^*$ we have 
\begin{enumerate}
    \item $(w,v) \in R$ if{f} $(w',v) \in R$, and
    \item $(v,w) \in R$ if{f} $(v,w') \in R$.
\end{enumerate}

It turns out that equivalence classes of $\autequiv$ define the coarsest partition onto which $R$ can be recognized in terms of $\kREC$:

\begin{lemma}\AP\label{lem:krec-characterization}
    For every "automatic" $R \subseteq \A^* \times \A^*$, $\autequiv$ has index at most $k$ if, and only if, $R$ is in $\kREC$.
\end{lemma}
\begin{proof}
    \proofcase{Left-to-right}
    Assume that $\autequiv$ has the equivalence classes $E_1, \dotsc, E_k$. Consider the set $P \subseteq \set{1, \dotsc, k}^2$ of all pairs $(i,j)$ such that there are $u_i \in E_i$ and $u_j \in E_j$ with $(u_i,u_j) \in R$. Define the $\kREC$ relation $R' = \bigcup_{(i,j) \in P} E_i \times E_j$. We claim that $R=R'$. 
    In fact, by definition of $\autequiv$, note that if there are $u_i \in E_i$ and $u_j \in E_j$ with $(u_i,u_j) \in R$, then $E_i \times E_j \subseteq R$. Hence, $R' \subseteq R$.
    On the other hand, for every pair $(u,v) \in R$ there is $(i,j) \in P$ such that $u \in E_i$, $v \in E_j$ implying $(u,v) \in R'$.
    Hence, $R \subseteq R'$.

    \proofcase{Right-to-left}
    If $R$ is a union of products of sets from the partition $E_1 \dcup \dotsb \dcup E_k = \A^*$, then every two elements of each $E_i$ are $\autequiv$-related, and thus $\autequiv$ has index at most $k$.
\end{proof}

We can then conclude that the definability problem for $\kREC$ is decidable. 

\begin{corollary}
    The $\kREC$-"definability problem" is decidable, for every $k > 0$.
\end{corollary}
\begin{proof}
    An "automatic relation" $R$ is in $\kREC$ if{f} $\autequiv$ has at most $k$ equivalence classes by \Cref{lem:krec-characterization}. 
    %
    In other words, an "automatic relation" $R$ is not in $\kREC$ if{f} the complement of $\autequiv$ contains a $(k+1)$-clique, which can be easily tested.
\end{proof}

The relation $\autequiv$ can also be used to characterize which automatic relations are definable in the class $\kPROD$.

\begin{lemma}
    An "automatic relation" $R$ is in $\kPROD$ if, and only if, $R=(A_1 \times B_1) \cup \dotsb \cup (A_k \times B_k)$ where each $A_i$ and $B_i$ is a union of equivalence classes of $\autequiv$.
\end{lemma}
\begin{proof}
    It suffices to show that for every equivalence class $E$ from $\autequiv$, if $A_1 \cap E \neq \emptyset$ then $R = ((A_1 \cup E) \times B_1) \cup \dotsb \cup (A_k \times B_k)$, and similarly for $B_1$. Assume $w \in A_1 \cap E$ and take any pair $(u,v) \in E \times B_1$. We show that $(u,v) \in R$. By definition of $\autequiv$, since $(w,v) \in R$ and $w \mathrel{\autequiv} u$, we have that $(u,v) \in R$.
\end{proof}

Again, this characterization allows us to show that definability in the class $\kPROD$ is decidable. 

\kproddefdec

\begin{proof}
    By brute force testing whether the "automatic relation" $R$ is equivalent to $(A_1 \times B_1) \cup \dotsb \cup (A_k \times B_k)$ for every possible $A_i,B_i$ which is a union of equivalence classes of $\autequiv$.
\end{proof}

\section{Discussion}

We have established, among other things, the undecidability of the
"$k$-regular colorability problem" for $k \geq 2$. Yet, little is known about
the "regular colorability problem".

\begin{conjecture}
    The $\REC$-"separability problem"---or, equivalently, the "regular colorability problem"---is undecidable.
\end{conjecture}

Beyond its decidability status, the structural properties of "regular colorability" evades us:

\begin{conjecture}
    Over "automatic graphs", the following notions are pairwise disjoint:
    \begin{enumerate}
        \item to be finitely "regular colorable",
        \item to be finitely colorable,
        \item not to contain unbounded cliques.
    \end{enumerate}
\end{conjecture}

Note that the implications (1) $\Rightarrow$ (2) $\Rightarrow$ (3) trivially hold.
Moreover, recall that while the "automatic tree" of \Cref{ex:tree-not-2-reg-colorable}
is not "2-regular colorable", it is "3-regular colorable" (it suffices to color $\varepsilon$ with a new color, and then color $a^p b^q$ by looking at the parity of $p-q$). Hence,
it does not prove that (2) $\not\Rightarrow$ (1).
Likewise, on arbitrary infinite graphs, we know that
there exists triangle-free graphs that are not finitely
colorable \cite{ungar_descartes54chromatic}---but we believe these
graphs not to be "automatic", and hence they would not prove that (3) $\not\Rightarrow$ (2).

Finally, observe that it is decidable to test whether an "automatic graph" has \emph{infinite} cliques \cite[Corollary 5.5]{KL10}. We conjecture that this property generalizes to unbounded cliques.

\begin{conjecture}
    The problem of whether an "automatic graph" has bounded cliques is decidable.
\end{conjecture}

\diego{What about a fixed alphabet? Do all results hold already for a binary alphabet?}
\remi{I think that "wf-RTM" with 2 letters and one tape are as expressive as Turing machines.
A configuration can be encoded using 3 letters (2 letters for the tape content, the same letters encode the current state, but we might need a special symbol for separating them).
Hence, I think that \Cref{thm:k-reg-col-undec} already holds for only a fixed alphabet (3 letters using this argument, maybe we can do better).}

\bibliographystyle{plainurl}
\bibliography{long,biblio,biblio1,references1}

\begin{thebibliography}{10}

\bibitem{thispaperMFCS}
Pablo Barcel\'o, Diego Figueira, and R\'emi Morvan.
\newblock Separating automatic relations.
\newblock In {\em 48th International Symposium on Mathematical Foundations of
  Computer Science (MFCS 2023)}. Leibniz International Proceedings in
  Informatics (LIPIcs), 2023.

\bibitem{ICALP19}
Pablo Barcel{\'{o}}, Chih{-}Duo Hong, Xuan~Bach Le, Anthony~W. Lin, and Reino
  Niskanen.
\newblock Monadic decomposability of regular relations.
\newblock In {\em International Colloquium on Automata, Languages and
  Programming (ICALP)}, pages 103:1--103:14, 2019.
\newblock \href {https://doi.org/10.4230/LIPIcs.ICALP.2019.103}
  {\path{doi:10.4230/LIPIcs.ICALP.2019.103}}.

\bibitem{BarceloLLW12}
Pablo Barcel{\'o}, Leonid Libkin, Anthony~Widjaja Lin, and Peter~T. Wood.
\newblock Expressive languages for path queries over graph-structured data.
\newblock {\em ACM Transactions on Database Systems (TODS)}, 37(4):31, 2012.
\newblock \href {https://doi.org/10.1145/2389241.2389250}
  {\path{doi:10.1145/2389241.2389250}}.

\bibitem{BLSS03}
Michael Benedikt, Leonid Libkin, Thomas Schwentick, and Luc Segoufin.
\newblock Definable relations and first-order query languages over strings.
\newblock {\em Journal of the ACM}, 50(5):694--751, 2003.
\newblock \href {https://doi.org/10.1145/876638.876642}
  {\path{doi:10.1145/876638.876642}}.

\bibitem{BGLZ22}
Pascal Bergstr{\"{a}}{\ss}er, Moses Ganardi, Anthony~W. Lin, and Georg
  Zetzsche.
\newblock Ramsey quantifiers over automatic structures: Complexity and
  applications to verification.
\newblock In Christel Baier and Dana Fisman, editors, {\em Annual Symposium on
  Logic in Computer Science (LICS)}, pages 28:1--28:14. {ACM}, 2022.
\newblock \href {https://doi.org/10.1145/3531130.3533346}
  {\path{doi:10.1145/3531130.3533346}}.

\bibitem{Berstel}
Jean Berstel.
\newblock {\em Transductions and Context-Free Languages}.
\newblock Teubner-Verlag, 1979.

\bibitem{blumensath2023MSO}
Achim Blumensath.
\newblock {Monadic Second-Order Model Theory}.
\newblock Preprint of a book., 2023.
\newblock Version of 2023-02-25.
\newblock URL: \url{https://www.fi.muni.cz/~blumens/MSO.pdf}.

\bibitem{BG00}
Achim Blumensath and Erich Gr{\"{a}}del.
\newblock Automatic structures.
\newblock In {\em Annual Symposium on Logic in Computer Science (LICS)}, pages
  51--62. {IEEE} Computer Society, 2000.

\bibitem{CCG06}
Olivier Carton, Christian Choffrut, and Serge Grigorieff.
\newblock Decision problems among the main subfamilies of rational relations.
\newblock {\em {RAIRO} -- Theoretical Informatics and Applications},
  40(2):255--275, 2006.
\newblock \href {https://doi.org/10.1051/ita:2006005}
  {\path{doi:10.1051/ita:2006005}}.

\bibitem{Choffrut-survey}
Christian Choffrut.
\newblock Relations over words and logic: {A} chronology.
\newblock {\em Bull. of the EATCS}, 89:159--163, 2006.

\bibitem{CCLP17}
Lorenzo Clemente, Wojciech Czerwi{\'n}ski, S{\l}awomir Lasota, and Charles
  Paperman.
\newblock Regular separability of {P}arikh automata.
\newblock In {\em ICALP}, pages 117:1--117:13, 2017.
\newblock \href {https://doi.org/10.4230/LIPIcs.ICALP.2017.117}
  {\path{doi:10.4230/LIPIcs.ICALP.2017.117}}.

\bibitem{CMRZZ17}
Wojciech Czerwi{\'n}ski, Wim Martens, Lorijn van Rooijen, Marc Zeitoun, and
  Georg Zetzsche.
\newblock A characterization for decidable separability by piecewise testable
  languages.
\newblock {\em Discret. Math. Theor. Comput. Sci.}, 19(4), 2017.

\bibitem{EM}
Calvin~C. Elgot and Jorge~E. Mezei.
\newblock On relations defined by generalized finite automata.
\newblock {\em IBM J. Res. Dev.}, 9(1):47--68, 1965.
\newblock \href {https://doi.org/10.1147/rd.91.0047}
  {\path{doi:10.1147/rd.91.0047}}.

\bibitem{FS93}
Christiane Frougny and Jacques Sakarovitch.
\newblock Synchronized rational relations of finite and infinite words.
\newblock {\em Theor. Comput. Sci.}, 108(1):45--82, 1993.
\newblock \href {https://doi.org/10.1016/0304-3975(93)90230-Q}
  {\path{doi:10.1016/0304-3975(93)90230-Q}}.

\bibitem{IKR02}
Hajime Ishihara, Bakhadyr Khoussainov, and Sasha Rubin.
\newblock Some results on automatic structures.
\newblock In {\em Annual Symposium on Logic in Computer Science (LICS)}, page
  235. {IEEE} Computer Society, 2002.
\newblock \href {https://doi.org/10.1007/11690634_22}
  {\path{doi:10.1007/11690634_22}}.

\bibitem{Kocher}
Chris K\"ocher.
\newblock {\em Analyse der Entscheidbarkeit diverser Probleme in automatischen
  Graphen}.
\newblock Unpublished manuscript, 2014.
\newblock URL:
  \url{https://people.mpi-sws.org/~ckoecher/files/theses/bsc-thesis.pdf}.

\bibitem{K16}
Eryk Kopczy{\'n}ski.
\newblock Invisible pushdown languages.
\newblock In {\em Annual Symposium on Logic in Computer Science (LICS)}, pages
  867--872, 2016.
\newblock \href {https://doi.org/10.1145/2933575.2933579}
  {\path{doi:10.1145/2933575.2933579}}.

\bibitem{DBLP:conf/ifipTCS/KuskeL08}
Dietrich Kuske and Markus Lohrey.
\newblock Hamiltonicity of automatic graphs.
\newblock In Giorgio Ausiello, Juhani Karhum{\"{a}}ki, Giancarlo Mauri, and
  C.{-}H.~Luke Ong, editors, {\em IFIP}, volume 273, pages 445--459. Springer,
  2008.

\bibitem{KL10}
Dietrich Kuske and Markus Lohrey.
\newblock Some natural decision problems in automatic graphs.
\newblock {\em J. Symb. Log.}, 75(2):678--710, 2010.
\newblock \href {https://doi.org/10.2178/jsl/1268917499}
  {\path{doi:10.2178/jsl/1268917499}}.

\bibitem{lecerf1963machines}
Yves Lecerf.
\newblock Machines de turing réversibles. récursive insolubilité en $n \in
  \mathbf{N}$ de l'équation $u = \theta^n u$, ou $\theta$ est un
  « isomorphisme de codes ».
\newblock {\em Comptes rendus hebdomadaires des séances de l'Académie des
  sciences}, 257:2597--2600, 1963.

\bibitem{LB16}
Anthony~W. Lin and Pablo Barcel\'{o}.
\newblock String solving with word equations and transducers: Towards a logic
  for analysing mutation {XSS}.
\newblock In {\em Annual Symposium on Principles of Programming Languages
  (POPL)}, pages 123--136. {ACM}, 2016.
\newblock \href {https://doi.org/10.1145/2837614.2837641}
  {\path{doi:10.1145/2837614.2837641}}.

\bibitem{DBLP:journals/dmtcs/LodingS19}
Christof L{\"{o}}ding and Christopher Spinrath.
\newblock Decision problems for subclasses of rational relations over finite
  and infinite words.
\newblock {\em Discret. Math. Theor. Comput. Sci.}, 21(3), 2019.
\newblock \href {https://doi.org/10.23638/DMTCS-21-3-4}
  {\path{doi:10.23638/DMTCS-21-3-4}}.

\bibitem{Morita2017}
Kenichi Morita.
\newblock {\em Reversible Turing Machines}, pages 103--156.
\newblock Springer Japan, Tokyo, 2017.
\newblock \href {https://doi.org/10.1007/978-4-431-56606-9_5}
  {\path{doi:10.1007/978-4-431-56606-9_5}}.

\bibitem{Nivat}
Maurice Nivat.
\newblock Transduction des langages de {C}homsky.
\newblock {\em Ann. Inst. Fourier}, 18:339--455, 1968.

\bibitem{PZ14}
Thomas Place and Marc Zeitoun.
\newblock Separating regular languages with first-order logic.
\newblock {\em Logical Methods in Computer Science (LMCS)}, 12(1), 2016.
\newblock \href {https://doi.org/10.2168/LMCS-12(1:5)2016}
  {\path{doi:10.2168/LMCS-12(1:5)2016}}.

\bibitem{Stearns}
Richard~Edwin Stearns.
\newblock A regularity test for pushdown machines.
\newblock {\em Information and Control}, 11(3):323--340, 1967.
\newblock \href {https://doi.org/10.1016/S0019-9958(67)90591-8}
  {\path{doi:10.1016/S0019-9958(67)90591-8}}.

\bibitem{ungar_descartes54chromatic}
Peter Ungar and Blanche Descartes.
\newblock $k$-{C}hromatic graphs without triangles.
\newblock {\em The American Mathematical Monthly}, 61(5):352--353, 1954.
\newblock \href {https://doi.org/10.2307/2307489} {\path{doi:10.2307/2307489}}.

\bibitem{Valiant75}
Leslie~G. Valiant.
\newblock Regularity and related problems for deterministic pushdown automata.
\newblock {\em Journal of the {ACM}}, 22(1):1--10, 1975.
\newblock \href {https://doi.org/10.1145/321864.321865}
  {\path{doi:10.1145/321864.321865}}.

\end{thebibliography}

\ifarxiv
    \appendix
    \clearpage
    \section{Missing Proofs}
\label{apdx-sec:missing-proofs}

\begin{proofappendix}{lem:incomp-is-automatic}{\incompisautomatic}
    By definition, the "incompatibility relation" $\incompGraph{R_1}{R_2}$ can be written as
	$R_{\neg\compL} \cup R_{\neg\compLpr} \cup R_{\neg\compR} \cup R_{\neg\compRpr}$, where:
	\begin{align*}
        R_{\neg\compL} &\defeq \big\{ (u,u') \in \A^* \times \A^* \;\big\vert\; \exists v \in \A^*,\; (u,v) \in R_1 \land (u',v) \in R_2 \big\} \text{,}\\
        R_{\neg\compLpr} &\defeq \big\{ (u,u') \in \A^* \times \A^* \;\big\vert\; \exists v \in \A^*,\; (u',v) \in R_1 \land (u,v) \in R_2 \big\} \text{,}\\
        R_{\neg\compR} &\defeq \big\{ (u,u') \in \A^* \times \A^* \;\big\vert\; \exists v \in \A^*,\; (v,u) \in R_1 \land (v,u') \in R_2 \big\} \text{, and}\\
        R_{\neg\compRpr} &\defeq \big\{ (u,u') \in \A^* \times \A^* \;\big\vert\; \exists v \in \A^*,\; (v,u') \in R_1 \land (v,u) \in R_2 \big\}
    \end{align*}
    Observe that starting from automata for $R_1$ and $R_2$, then for each
	of the relation $R_{\neg\compL}$, $R_{\neg\compLpr}$, $R_{\neg\compR}$ or $R_{\neg\compRpr}$, we can build an automaton recognizing them
	using a product construction, which can be implemented in polynomial time.
    It then follows that we can build a polynomial automaton recognizing
    $\incompGraph{R_1}{R_2}$.
\end{proofappendix}

\begin{proofappendix}{lem:reachable-regularity}{\reachableregularity}
    By reduction from the halting problem for deterministic and "reversible" TMs, which is undecidable by \Cref{prop:halting-problem-detrevTM}. Given a deterministic and "reversible" TM $T$ (running on the empty input), consider the TM $T'$ where every time there is a transition $(u, p, v) \to (u', q, v')$ from configuration $c$ to configuration $c'$ in $T$,
	simulate this transition in $T'$---$a$'s should be treated as blank symbols---,
	and then rewrite $a^n b^n$ into $a^{n+1}b^{n+1}$.
	When $T$ writes on a blank symbol that was actually a $a$ in $T'$,
	we must also add an extra $a$ (to account for the one that was overwritten):
	this case is depicted \Cref{fig:reachable-regularity}.
	Moreover, when $T$ deletes a symbol at the end of the tape,
	we must shift the $a^n b^n$ prefix. This can be done by replacing the blank
	with an $a$, the last $a$ with a $b$, and deleting the last $b$.
	
    Observe that $T'$ is a "wf-RTM":
    \begin{enumerate}
        \item the "initial configuration" $(\bot,q_0,\bot)$ has no predecessor;
        \item it is deterministic and co-deterministic:
        \begin{itemize}
            \item every configuration inside a path $(u, q, v a^n b^n) \xrightarrow{*} (u, q, v a^{n+1} b^{n+1})$ 
            has, by definition, exactly in- and out-degree one;
            \item every configuration of the form $(u, p, v a^n b^n)$ has as many predecessors [resp.  
                successors] in $T'$ as $(u,q,v)$ in $T$, namely one since $T$ was assumed to be
                deterministic and "reversible";
        \end{itemize}
        \item it has no infinite descending chain since $\N$ is well-founded.
    \end{enumerate}
    Moreover $T'$ has no cycle,
    so if $T$ is halting (on an empty input) then the set of "reachable configurations" of $T'$ is finite (since it is a "wf-RTM") and thus regular. If $T$ is not halting, the set of "reachable configurations" of $T'$ is infinite and its projection onto $\set{a,b}$ is an infinite set of words of the form $a^{n} b^{n'}$ where $n-2 \leq n' \leq n+2$. Hence, since regular languages are closed under homomorphic images, the "reachable configurations" of $T'$ cannot be regular.
\end{proofappendix}
    \section{Incompatibility Graph}
\label{apdx-sec:incompatibility}

\begin{example}
   \AP\label{ex:equal_length_plusone}
   Let $\A = \{a,b\}$, $R_1$ be the equal-length relation,
   and
   \[
       R_2 = \{(u, ua) \mid u \in \A^*\} \cup \{(u, ub) \mid u \in \A^*\}.
   \]
   Then, $u$ is "incompatible" with $u'$ if $|u| = |u'|+1$ (this is given by \compL~or \compR),
   or if $|u'| = |u|+1$ (this is given by \compLpr~or \compRpr).
   $R_2$ and the "incompatibility graph" are depicted in \Cref{fig:equal_length_plusone}.
   
   Note that $R_1$ and $R_2$ are separable by the "recognizable" relation
   $S$ consisting of all pairs $(u,v)$ such that $|u|$ and $|v|$ have the same parity.
   Moreover, $\incompGraph{R_1}{R_2}$ is "2-regular colorable", the two colors being
   the words of even and odd length.
\end{example}

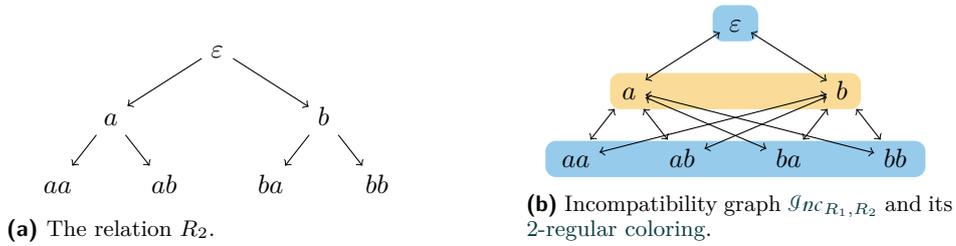
\begin{figure}[htb]
   \centering%
   \begin{subfigure}[b]{.4\textwidth}
       \centering
       \begin{tikzpicture}[>={Classical TikZ Rightarrow},
	level distance=.9cm,
	level 1/.style={sibling distance=2.8cm},
	level 2/.style={sibling distance=1.4cm},
	level 3/.style={sibling distance=.7cm},
	edge from parent/.style={draw,->}]
	\node (t) {$\varepsilon\vphantom{b}$}
		child {node {$a\vphantom{b}$}
			child {node {$aa\vphantom{b}$}
			}
			child {node {$ab$}
			}
		}
		child {node {$b$}
			child {node {$ba$}
			}
			child {node {$bb$}
			}
		};
\end{tikzpicture}
       \caption{The relation $R_2$.}
   \end{subfigure}
   \hspace{1cm}
   \begin{subfigure}[b]{.4\textwidth}
       \centering
       \begin{tikzpicture}[>={Classical TikZ Rightarrow},
	level distance=.9cm,
	level 1/.style={sibling distance=2.8cm},
	level 2/.style={sibling distance=1.4cm},
	level 3/.style={sibling distance=.7cm},
	edge from parent/.style={}]
	\fill[rounded corners, fill=cBlue, opacity=.5]
		(-2.5,-2.05) rectangle (2.5,-1.57)
		(-.3,-0.25) rectangle (.3,0.23);
	\fill[rounded corners, fill=cYellow, opacity=.5]
		(-1.65,-1.15) rectangle (1.65,-0.67);
	\node (eps) {$\varepsilon\vphantom{b}$}
		child {node (a) {$a\vphantom{b}$}
			child {node (aa) {$aa\vphantom{b}$}
			}
			child {node (ab) {$ab$}
			}
		}
		child {node (b) {$b$}
			child {node (ba) {$ba$}
			}
			child {node (bb) {$bb$}
			}
		};
	\draw[<->] (eps) edge (a)
		(eps) edge (b)
		(a) edge (aa)
		(a) edge (ab)
		(a) edge (ba)
		(a) edge (bb)
		(b) edge (aa)
		(b) edge (ab)
		(b) edge (ba)
		(b) edge (bb);
\end{tikzpicture}
       \caption{Incompatibility graph $\incompGraph{R_1}{R_2}$ and its "2-regular coloring".}
   \end{subfigure}
   \caption{\AP\label{fig:equal_length_plusone}Automatic graphs of \Cref{ex:equal_length_plusone}, restricted to words of length at most 2.}
\end{figure}
\fi

\end{document}